\documentclass[journal]{IEEEtran}
\ifCLASSINFOpdf
\else
\fi
\usepackage{amssymb}
\usepackage{amsmath}
\usepackage{epsfig}
\usepackage{xcolor}
\usepackage{lscape}
\usepackage{graphicx}
\usepackage{graphicx}
\usepackage{lscape}
\usepackage{subcaption}
\usepackage[utf8]{inputenc}
\usepackage[english]{babel}
\usepackage{amssymb,amsmath,amsthm}
\usepackage{float}
\usepackage{xcolor}
\usepackage{url}
\newtheorem{theorem}{Theorem}

\usepackage{mathtools}
\usepackage{verbatim} 

\newcommand{\bE}{\mathbb{E}}
\newcommand{\tp}{^{\intercal}}

\begin{document}
%
\title{Adaptive Learning with Artificial Barriers Yielding Nash Equilibria in General Games}
%
%
%

\author{Ismail~Hassan, 
    and~B.John~Oommen,~\IEEEmembership{Life~Fellow,~IEEE},
     Anis~Yazidi,~\IEEEmembership{Senior Member,~IEEE}
\thanks{Ismail Hassan, Author's status: {\it Assistant Professor}.
This author can be contacted at: Oslo Metropolitan University, Department of Computer Science, Pilestredet 35, Oslo, Norway.  E-mail:
{\tt ismail@oslomet.no}.},
\thanks{B.~John~Oommen, {\it Chancellor's Professor} ; {\it Life Fellow:
IEEE} and {\it Fellow: IAPR}. The work of this author was partially supported by NSERC, the Natural Sciences and Engineering Council of Canada. This author can be contacted at:
School of Computer Science, Carleton University, Ottawa, Canada :
K1S 5B6. This author is also an {\em Adjunct Professor} with the
University of Agder in Grimstad, Norway. E-mail address: \texttt{oommen@scs.carleton.ca}.}, 
\thanks{Anis Yazidi, Author's status: {\it Professor}.
This author can be contacted at: Oslo Metropolitan University, Department of Computer Science, Pilestredet 35, Oslo, Norway.  E-mail:
{\tt anis.yazidi@oslomet.no}.}
}

\maketitle

\begin{abstract}
Artificial barriers in Learning Automata (LA) is a powerful and yet under-explored concept although it was first proposed in the 1980s \cite{Oommen1986}. Introducing  artificial non-absorbing barriers makes the LA schemes resilient to being trapped in absorbing barriers, a phenomenon which is often referred to as lock in probability leading to an exclusive choice of one action after convergence.  Within the field of LA and reinforcement learning in general, there is a sacristy
of theoretical works and applications of schemes with artificial barriers. In this paper, we devise a LA with artificial barriers for solving a general form of stochastic bimatrix game. Classical LA  systems  possess properties of absorbing barriers and they are a powerful tool in game theory and were shown to converge to game's of Nash equilibrium under limited information \cite{sastry1994decentralized}. However, the stream of works in LA for solving game theoretical problems can merely solve the case where the Saddle Point of the game exists in a pure strategy and fail to reach mixed Nash equilibrium when no Saddle Point exists for a pure strategy.
In this paper, by resorting to the powerful concept of artificial barriers, we suggest a LA that converges to an optimal mixed Nash equilibrium even though there may be no Saddle Point when a pure strategy is invoked.
Our deployed scheme is of Linear Reward-Inaction ($L_{R-I}$) flavor which is originally an absorbing LA scheme, however, we render it non-absorbing by introducing artificial barriers in an elegant and natural manner, in the sense that that the well-known legacy $L_{R-I}$ scheme can be seen as an instance of our proposed algorithm for a particular choice of the barrier. Furthermore, we present an $S$ Learning version of our LA with absorbing barriers that is able to handle $S$-Learning environment in which the feedback is continuous and not binary as in the case of the $L_{R-I}$. 

Reward-$\epsilon$Penalty ($L_{R-\epsilon P}$) scheme proposed by Lakshmivarahan and Narendra \cite{lakshmivarahan1982learning}  almost four decades ago, is the only LA scheme that was shown to converge to the optimal mixed Nash equilibrium when no Saddle Point exists in pure strategy, and the proofs were limited to only zero-sum games.
In this paper, we tackle a general game as opposed to the particular case of zero-sum games proposed by Lakshmivarahan and Narendra in \cite{lakshmivarahan1982learning}, and provide a novel scheme and proofs characterizing its behavior. Furthermore, we provide experimental results that are in line with our theoretical findings.



\end{abstract}


%
\IEEEpeerreviewmaketitle

\section{Introduction}
\label{sec:Introduction}
Narendra and Thathachar first presented the term Learning Automata (LA) in their 1974 survey \cite{narendra74}. LA consists of an adaptive learning agent interacting with a stochastic environment with incomplete information. Lacking prior knowledge, LA attempts to determine the optimal action to take by first choosing an action randomly and then updating the action probabilities based on the reward/penalty input that the LA receives from the environment. This process is repeated until the optimal action is, finally, achieved. The LA update process can be described by the learning loop shown in Figure \ref{fig:LAloop}.

Formally, a LA is defined by the mean of a quintuple $\langle {A,
B, Q, F(.,.), G(.)} \rangle$, where the elements of the quintuple are defined term by term  as:

\begin{enumerate}
\item   $A=\{ {\alpha}_1, {\alpha}_2, \ldots, {\alpha}_r\}$ gives the set of actions available to the LA, while ${\alpha}(t)$
    is the action selected at time instant $t$ by the LA. Note that the LA selects one action at a time, and the selection is sequential.

\item   $B = \{{\beta}_1, {\beta}_2, \ldots, {\beta}_m\}$ denotes the set of possible input values that the LA can receive. ${\beta}(t)$ denotes the input at time instant $t$ which is a form of feedback.

\item   $Q=\{q_1, q_2, \ldots, q_s\}$ represents the states of the LA where $Q(t)$ is the state at time instant
$t$.

\item   $F(.,.): Q \times B \mapsto  Q$ is the transition function at time $t$, such that, $q(t+1)=F[q(t),
{\beta}(t)]$. In simple terms, $F(.,.)$  returns the next state of the LA at time instant $t+1$ given the current state and the input from the environment both at time $t$ using either a
deterministic or a stochastic mapping.

\item   $G(.)$ defines {\it output function},  it represents  a mapping $G: Q \mapsto  A$ which determines the action of the LA as a function of the state. 

\end{enumerate}

The Environment, $E$ is characterized by :

\begin{itemize}
    \item $C=\{c_1 , c_2 ,\ldots, c_r\}$ is a set of penalty
probabilities, where $c_i \in C$ corresponds to the penalty of action ${\alpha}_i$.
\end{itemize}

\begin{figure}[htp!]
\centering
\includegraphics [scale=1] {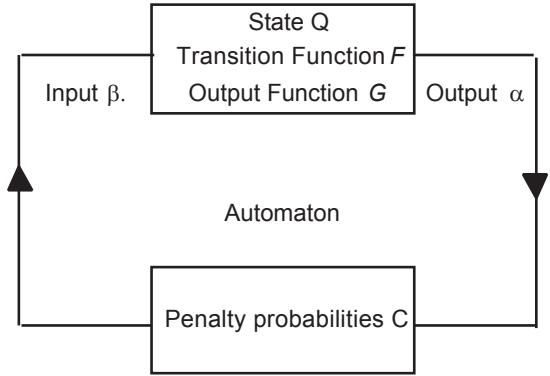}
\caption{LA interacting with the environment.}
\label{fig:LAloop}
\end{figure}

\noindent
{\bf Learning Automata}: Research into LA over the past four decades is extensive, leading to the proposal of various types throughout the years. LA are mainly characterized as being Fixed Structure Learning Automata (FSLA) or Variable Structure Learning Automata (VSLA). In FSLA, the probability of the transition from one state to another state is fixed and the action probability of any action in any state is also fixed. Early research into LA centered around FSLA. Early pioneers in LA such as Tsetlin, Krylov, and Krinsky \cite{Tse73} proposed several examples  of this types of automata. The research into LA moved gradually towards VSLA. Introduced by Varshavskii and Vorontsova in the early 1960's \cite{VV63}, VSLA has transition and output functions that evolve as the learning process continues \cite{Oommen1986}. The state transitions or the action probabilities are updated at every time step.

\noindent
{\bf Continuous or Discretized}: VSLA can also be defined as being Continuous or Discretized depending on the values that the action probabilities can take. In continuous LA, the action probabilities can take any value in the interval $[0,1]$. The drawback with continuous LA is that they {\em approach} a goal but never {\em reach} there and they have a slow rate of convergence. The concept of discretization was introduced in the 1980s to address the shortcomings of continuous LA. The proposed mitigation increased LA speeds of convergence \cite{ Lanctot1992, Oommen1990a} by permitting an action probability that was close enough to zero or unity, jump to that end point in a single step. The method employed by the authors constrained the action selection probability to be one of a finite number of values in the interval $[0,1]$. 

\noindent
{\bf Ergodic or Absorbing}: Depending on their Markovian properties, LA can further be classified as either ergodic or equipped with characteristics of absorbing barriers. In an ergodic LA system, the final steady state does not depend on the initial state. In contrast to LA with absorbing barriers, the steady state depends on the initial state and when the LA converges, it gets locked into an absorbing state.

Absorbing barrier VSLA are preferred in static environments, while ergodic VSLA are suitable for dynamic environments.

\noindent
{\bf Environment types}:
The feedback from the LA  ${\beta}(t)$ is  a scalar that falls in the interval $[0,1]$. If the feedback is binary, meaning 0 or 1, then the Environment is called P-type. Whenever the feedback is a discrete values, we call the environment $Q-$type. In the third case where the feedback is any real number in the interval $[0,1]$, we call the environment as $S-$type. Traditionally the schemes that handle $P$-type environment operate with stochastic environments where the feedback is a a realization from a Bernoulli process, or at least in the proofs since the penalty and rewards do not really need in general to obey some statistical laws and can be for instance generated by an adversary. The scheme that handle $P$ types such as the famous  $L_{R-I}$ and a large class of fixed structured LA, do not automatically handle $S$-types and $Q-$types. Mason presented in \cite{mason1973optimal} an LA scheme to handle $S$-learning environment. In this paper, we generalize our scheme to the $S$-type environment and provide the corresponding update scheme.

\noindent
{\bf LA with Artificially Absorbing Barriers}: LA with artificially absorbing barrier were introduced in the 1980s. In \cite{Oommen1986}, Oommen  turned a discretized ergodic scheme into an absorbing one by introducing an artificially absorbing barrier that forces the scheme to converge to one of the absorbing barriers. Such a modification led to the advent of new LA families with previously unknown behavior.
For instance, the $ADL_{R-P}$ and $ADL_{I-P}$ are absorbing schemes that are the result of the introducing absorbing barriers to their counterparts original corresponding schemes. Those absorbing scheme  were shown to be $\epsilon$-optimal in all random environments.

\noindent
{\bf Applications of LA}: LA had been utilized in many applications over the years. Recent applications of LA include achieving fair load balancing based on two-time-scale separation paradigm \cite{ismail2020}, detection of malicious social bots \cite{ranjan2020}, secure socket layer certificate verification \cite{krishna2014secure}, a protocol for intrusion detection \cite{fathinavid2012}, efficient decision making mechanism for stochastic nonlinear resource allocation \cite{yazidi2019two}, link prediction in stochastic social networks \cite{moradabadi2018link}, user behavior analysis-based smart energy management for webpage ranking \cite{makkar2018user}, and resource selection in computational grids \cite{enami2019resource}, to mention a few.

\noindent
{\bf LA Applied to Game Theory}: 
Studies on strategic games with LA
were focused mainly on traditional $L_{R-I}$ which is desirable to use as it can yield Nash equilibrium in pure strategies \cite{sastry1994decentralized}. Although other ergodic schemes such as $L_{R-P}$  were used in games \cite{viswanathan1974games} with limited information, they did not gain popularity at least when it comes to applications due to their inability to converge to Nash equilibrium.
LA has found numerous applications in game theoretical applications such as sensor fusion without knowledge of the ground truth
\cite{yazidi2020solving}, 
for distributed power control in wireless networks and more particurly NOMA
\cite{rauniyar2020reinforcement},
optimization of cooperative tasks
\cite{zhang2020learning}, 
for content placement in cooperative caching \cite{yang2020learning},
congestion control in Internet of Things
\cite{gheisari2019cccla},
 QoS satisfaction in autonomous mobile edge computing
\cite{apostolopoulos2018game}, opportunistic spectrum access 
\cite{cao2018distributed}
  scheduling domestic shiftable loads in smart grids
 \cite{thapa2017learning},
     anti-jamming channel selection algorithm for interference mitigation
\cite{jia2017distributed},
relay selection  in vehicular ad-hoc networks
  \cite{tian2017self} etc.

\noindent
{\bf Objective and Contribution of this paper}: In this paper, we propose an algorithm addressing bimatrix games which is a more general version of the zero-sum game treated in \cite{lakshmivarahan1982learning}.

First we consider  a stochastic game where the outcomes are binary in our case, which are either a reward or a penalty. The reward probabilities are given by corresponding payoff matrix of each player. The game we treat is of limited information which is a flavor of games often treated in LA. In such game, each player only observes the outcome of his action in the form of a reward or penalty without observing the action chosen by the other player. The player might not be even aware that he is playing against an opponent player.
 By virtue of the design principles of our scheme, at each round of the repetitive game, the players revise their strategies upon receiving a reward while maintain their strategies unchanged upon receiving a penalty. This is in concordance with the Linear Reward-Inaction, $L_{R-I}$ paradigm. 
 Please note that this is radically different from the paradigm by Lakshmivarahan \cite{lakshmivarahan1982learning} where the players always revise their strategies at each round, where the magnitude of the adjustment in the probabilities of the action depend only on whether a reward or penalty is received at every time instant.
 
Furthermore, we provide an extension of our scheme to handle $S$-learning environment where the feedback is not binary but rather continuous. The informed reader will notice that our main focus is on the case of $P$-type environment, bthe sake of brevity and due to space limitations while we give enough exposure and attention related to the $S$-type environment.
The remainder of this article is organized as follows. In Section \ref{sec:GameModel}, we present the game model for both $P-$type environments and $S-$type environments.
In Section \ref{sec:scheme_L_R_I}, we introduce our devised $L_{R-I}$ with artificial barriers for handling $P-$type environments. In  Section \ref{sec:scheme_S_learning}, we present the $S-$ LA scheme with absorbing barriers for handling the general cases of $S-$type environments. The experimental results related to the $L_{R-I}$ are presented in Section \ref{sec:simualations_L_RI}
while some experiments of $S-$ LA for handling   $S-$ type environments are given in the Appendix \ref{sec:appendix_B}.
 

\section{The Game Model}
\label{sec:GameModel}

In this section, we formalize the game model that is being investigated. Let $P(t)= \begin{bmatrix} p_1(t) & p_2(t)\end{bmatrix}\tp$ denote the mixed strategy of player $A$ at time instant $t$, where $p_1(t)$ accounts for the probability of adopting strategy $1$ and, conversely, $p_2(t)$ stands for the probability of adopting strategy $2$. Thus, $P(t)$ describes the distribution over the strategies of player $A$. Similarly, we can define the mixed strategy of player $B$ at time $t$ as $Q(t)= \begin{bmatrix}q_1(t) & q_2(t)\end{bmatrix}\tp$.
 The extension to more than two actions per player is straightforward following the method analogous to what was used by Papavassilopoulos \cite{papavassilopoulos1989learning}, which extended the work of Lakshmivarahan and Narendra \cite{lakshmivarahan1982learning}.

Let $\alpha_A(t) \in \{1,2\}$ be the action chosen by player $A$ at time instant $t$ and $\alpha_B(t) \in \{1,2\}$ be the one chosen by player $B$, following the probability distributions $P(t)$ and $Q(t)$, respectively. The pair $(\alpha_A(t),\alpha_B(t)$) constitutes the joint action at time $t$, and are pure strategies. Specifically, if $(\alpha_A(t),\alpha_B(t))=(i,j)$, the probability of reward for player $A$ is determined by $r_{ij}$ while that of player $B$ is determined by $c_{ij}$. Player $A$ 
is in this case the row player while player $B$ is the column player.

When we are operating in the $P$-type mode, the game is defined by two payoff matrices, $R$ and $C$ describing the reward probabilities of player $A$ and player $B$ respectively:

\begin{equation}
R=\begin{pmatrix}r_{11} & r_{12} \\ r_{21} &r_{22}\end{pmatrix},
\label{ref:R_matrix}
\end{equation}
\noindent

and the matrix $C$

\begin{equation}
C=\begin{pmatrix}c_{11} & c_{12} \\ c_{21} &c_{22}\end{pmatrix},
\label{ref:C_matrix}
\end{equation}
\noindent

where, as aforementioned,  all the entries of both matrices are probabilities.

In the case where the environment is a $S$-model type, the latter two matrices are deterministic and describe the feedback as a scalar in the interval $[0,1]$. For instance, if we operate in the $S$-type environment, the feedback when both players choose their respective first actions will be the scalar $c_11$ for player $A$ and not Bernoulli feedback such in the previous case of $P$-type environment. It is possible also to consider $c_11$ as stochastic continuous variable with mean  $c_11$ and which realization in $c_11$, however, for the sake of simplicity we consider $c_11$, and consequently $C$ and $R$ as deterministic. The asymptotic convergence proofs for the $S-$type environment will remain valid independently of whether $C$ and $R$ are deterministic  or whether they are obtained from a distribution with support in the interval $[0,1]$ and with their means defined by the matrices.

Independently of the environment type, whether it is $P-$type or $S-$type environments, we have three cases to be distinguished for equilibria:
\begin{itemize}

\item Case 1:  if $(r_{11}-r_{21})(r_{12}-r_{22})<0$, 
$(c_{11}-c_{12})(c_{21}-c_{22})<0$ and 
$(r_{11}-r_{21})(c_{11}-c_{12})<0$, there is just one mixed equilibrium.
\item Case 2: if $(r_{11}-r_{21})(r_{12}-r_{22})>0$ or
$(c_{11}-c_{12})(c_{21}-c_{22})>0$, then there is just one pure equilibrium since there is one player at least who has a dominant strategy.
\item Case 3:  if $(r_{11}-r_{21})(r_{12}-r_{22})<0$, 
$(c_{11}-c_{12})(c_{21}-c_{22})<0$ and 
$(r_{11}-r_{21})(c_{11}-c_{12})>0$, there are two pure equilibria and one mixed equilibrium.

\end{itemize}

In strategic games, Nash equilibria are equivalently called the ``Saddle Points'' for the game. Since the outcome for a given joint action is stochastic, the game is of stochastic genre.

\section{Game Theoretical LA Algorithm based on the $L_{R-I}$ with Artificial Barriers}
\label{sec:scheme_L_R_I}

In this section, we shall present our $L_{R-I}$ with artificial barriers that is devised specially for the $P$-type environments.

\subsection{Non-Absorbing Artificial Barriers}
\label{sec:Non-AbsorbingBarriers}

As we have seen above from surveying the literature, an originally ergodic LA can be rendered absorbing by operating a change in its end states. However, what is unkown in the literature is a scheme which is  originally  absorbing can be rendered ergodic. In many cases, this can be achieved by making the scheme behave according to to the absorbing scheme rule over the probability simplex and pushing the probability back inside the simplex whenever the scheme approaches absorbing barriers. Such a scheme is novel in the field of LA and its advantage is that the strategies avoids being absorbed in non-desirable absorbing barriers. Further, and interestingly, by countering the absorbing barriers, the scheme can migrate stochastically towards a desirable mixed strategy. Interestingly, as we will see later in the paper, even if the optimal strategy corresponds to an absorbing barrier the scheme will approach it. Thus, the scheme converges  to mixed strategies whenever they correspond to optimal strategies while approaching the absorbing states whenever they are the optimal strategies.
    We shall give the details of our devised scheme in the next section which enjoy the above mentioned properties.

\subsection{Non-Absorbing Game Playing}
\label{sec:Non-AbsorbingGames}

At this juncture, we shall present the design of our proposed LA scheme together with some theoretical results demonstrating that it can converge to the Saddle Points of the game even if the Saddle Point is a {\em mixed} Nash equilibrium. Our solution presents a new variant of the the $L_{R-I}$ scheme, which is made rather ergodic by modifying the update rule in a general form which makes the original $L_{R-I}$ with absorbing barriers corresponding to the corners of the simplex an instance of the latter general scheme for a particular choice of parameters of the scheme. The proof of convergence is based on Norman's theory for learning processes characterized by small learning steps \cite{Narendra2012, Norman1972}.

We introduce $p_{max}$ as the artificial barrier which is a real value close to 1. Similarly, we introduce $p_{min}=1-p_{max}$ which corresponds to the lowest value any action probability can take. In order to enforce the constraint that the probability of any action for both players remains within the interval ${ [}p_{min}, p_{max} {]}$ one should start by choosing  initial  values of $p_1(0)$ and $q_1(0)$ in the same interval, and further resorting to updates rules that ensure that each update keeps the probabilities within the same interval.

If the outcome from the environment is a reward at a time $t$ for action $i \in \{1,2\}$, the update rule is given by:
\begin{equation}
\label{eq:algorithmGain_LRI}
\begin{aligned}
p_i(t+1) & =   p_i(t)+\theta (p_{max}-p_i(t))\\
p_s(t+1) & =    p_s(t)+\theta (p_{min}-p_s(t)) &\textrm{for } \quad s\ne i.
\end{aligned}
\end{equation}
\noindent
where $\theta$ is a learning parameter.
The informed reader observes that the update rules coincides with the classical  $L_{R-I}$ except that  $p_{max}$ replaces unity for updating $p_i(t+1)$ and that $p_{min}$ replaces zero for updating $p_s(t+1)$.

Following the Inaction principle of the $L_{R-I}$, whenever the player receives a penalty, its action probabilities are kept unchanged which is formally given by:

\begin{equation}
\label{eq:algorithmLoss}
\begin{aligned}
p_i(t+1) & =    p_i(t)\\
p_s(t+1) & =   p_s(t) &\textrm{for} \quad s\ne i.
\end{aligned}
\end{equation}

The update rules for the mixed strategy $q(t+1)$ are defined in a similar fashion. We shall now move to a theoretical analysis of the convergence properties of our proposed algorithm for solving a strategic game. In order to denote the optimal Nash equilibrium of the  game we use the pair $(p_{\mathrm{opt}},q_{\mathrm{opt}})$.

We also should distinguish detail of the equilibrium according to the entries in the payoff matrices $R$ and $C$ for Case 1.
\paragraph{Case 1: Only One Mixed Nash Equilibrium Case (No Saddle Point in pure strategies)}
The first case depicts the situation where no Saddle Point exists in pure strategies. In other words, the only Nash equilibrium is a mixed one.
Based on the fundamentals of Game Theory, the optimal mixed strategies can be shown to be the following:

$$p_{\rm opt}= \frac{c_{22}-c_{21}}{L^{\prime}}, \quad q_{\rm opt}=\frac{r_{22}-r_{12}}{L},$$
\noindent
where
$L=(r_{11}+r_{22})-(r_{12}+r_{21})$
and
$L^{\prime}=(c_{11}+c_{22})-(c_{12}+c_{21})$.
This case can be divided into two sub-cases.
The first sub-case given by:
\begin{equation}
    r_{11} > r_{21},  r_{12} < r_{22}; c_{11} < c_{12},  c_{21} > c_{22},\label{eq:sub-cond1}
\end{equation}

The second sub-case given by:
\begin{equation}
    r_{11} < r_{21},  r_{12} > r_{22}; c_{11} > c_{12},  c_{21} < c_{22},  
    \label{eq:sub-cond2}
\end{equation}

Let the vector $X(t) = \begin{bmatrix} p_1(t) & q_1(t) \end{bmatrix}\tp$. We resort to the the notation $\Delta X(t) = X(t+1) - X(t)$. For denoting the conditional expected value operator we use the nomenclature $\bE[\cdot | \cdot]$. 
Using those notations, we introduce the next theorem of the article.
\begin{theorem}
\label{thm:expectedValue}
Consider a two-player game with a payoff matrices as in Eq. \eqref{ref:R_matrix} and Eq. \eqref{ref:C_matrix},  and a learning algorithm defined by equations Eq. \eqref{eq:algorithmGain_LRI} and Eq. \eqref{eq:algorithmLoss} for both players $A$ and $B$, with learning rate $\theta$.
Then, $E[\Delta X(t) | X(t)]=\theta W(x)$ and for every $\epsilon>0$, there exists a unique stationary point $X^*=\begin{bmatrix}p_1^* &q_1^*\end{bmatrix}\tp$ satisfying:
\begin{enumerate}
    \item $W(X^*)=0$;
    \item $|X^*-X_{\rm opt}|<\epsilon$.
\end{enumerate}
\end{theorem}

\begin{proof}
We start by fist  computing the conditional expected value of the increment $\Delta X(t)$:
\begin{align*}
E[\Delta X(t) | X(t)]
&= E[X(t+1)-X(t)\vert X(t)]\\
&=\begin{bmatrix}
E[p_1(t+1)-p_1(t)\vert X(t)] \\ E[q_1(t+1)-q_1(t)\vert X(t)])
\end{bmatrix}\\
&=\theta\begin{bmatrix}
W_1(X(t)) \\ W_2(X(t))
\end{bmatrix}\\
&=\theta W(X(t)),
\end{align*}
where the above format is possible since all possible updates share the form $\Delta X(t) = \theta W(t)$, for some $W(t)$, as given in Eq. \eqref{eq:algorithmGain_LRI}.

For ease of notation, we drop the dependence on $t$ with the implicit assumption that all occurrences of $X$, $p_1$ and $q_1$ represent $X(t)$, $p_1(t)$ and $q_1(t)$ respectively. $W_1(x)$ is then:

\begin{equation}
\resizebox{0.9\hsize}{!}{$
\begin{aligned}
    W_1({X}) = \ \ &p_1 q_1 r_{11} (p_{max} - p_1) + p_1 (1-q_1) r_{12} (p_{max} - p_1)\\
    + &(1-p_1) q_1 r_{21} (p_{min} - p_1) \\
    + &(1-p_1) (1-q_1) r_{22} (p_{min} - p_1) \\
    = \ \ & p_1 \left [ q_1 r_{11} + (1-q_1) r_{12} \right ] (p_{max} - p_1) \\
    + & (1-p_1) \left [ q_1 r_{21} + (1-q_1) r_{22} \right ] (p_{min} - p_1) \\
    = \ \ &  p_1 (p_{max} - p_1) D_1^A(q_1) + (1-p_1) (p_{min} - p_1) D_2^A(q_1),  \\
\end{aligned}
$}
\end{equation}
where,
\begin{equation}
D_{1}^{A}(q_1)=q_1 r_{11}+(1-q_1)r_{12}
\end{equation}
\begin{equation}
D_{2}^{A}(q_1)=q_1 r_{21}+(1-q_1)r_{22}.
\end{equation}
By replacing $p_{max} = 1-p_{min}$ and rearranging the expression we get:
\begin{align*}
W_{1}({X})=& \ \ p_1(1-p_1)D_1^A(q_1) - p_1 p_{min} D_1^A(q_1) \\
+ & (1-p_1)p_{min} D_2^A(q_1) - p_1(1-p_1)D_2^A(q_1) \\
=& \ \ p_1 (1-p_1) \left [ D_{1}^{A}(q_1)-D_{2}^{A}(q_1)\right ] \\& -p_{min} \left [ p_1 D_{1}^{A}(q_1)-(1-p_1) D_{2}^{A}(q_1)\right ].
\end{align*}

Similarly, we can get
\begin{equation}
\resizebox{0.9\hsize}{!}{$
\begin{aligned}
    W_2({X}) = \ \ &q_1 p_1 c_{11} (p_{max} - q_1) +\\& q_1 (1-p_1) c_{21} (p_{max} - q_1)\\
    + &(1-q_1) p_1 c_{12} (p_{min} - q_1) +\\& (1-q_1) (1-p_1) c_{22} (p_{min} - q_1) \\
    = \ \ & q_1 \left [ p_1 c_{11} + (1-p_1) c_{21} \right ] (p_{max} - q_1) \\
    + & (1-q_1) \left [ p_1 c_{12} + (1-p_1) c_{22} \right ] (p_{min} - q_1) \\
    = \ \ &  q_1 (p_{max} - q_1)  D_1^B(p_1)  + \\& (1-q_1) (p_{min} - q_1)  D_2^B(p_1)   \\
\end{aligned}$}
\end{equation}
where
\begin{equation}
D_{1}^{B}(p_1)=p_1 c_{11}+(1-p_1)c_{21}
\end{equation}
\begin{equation}
D_{2}^{B}(p_1)=p_1 c_{12}+(1-p_1)c_{22}.
\end{equation}
By replacing $p_{max} = 1-p_{min}$ and rearranging the expression we get:
\begin{equation}
\resizebox{\hsize}{!}{$
\begin{aligned}
W_{2}({X})=& \ \ q_1(1-q_1)(1-D_1^B(p_1)) - q_1 p_{min} D_1^B(p_1) \\
+ & (1-q_1)p_{min} D_2^B(p_1) - q_1(1-q_1) D_2^B(p_1) \\
=& q_1 (1-q_1) \left [ D_{1}^{B}(p_1)-D_{2}^{B}(p_1)\right ] - \\& p_{min} \left [ q_1 D_{1}^{B}(p_1)-(1-q_1)D_{2}^{B}(p_1)\right ].
\end{aligned}
$}
\end{equation}

We need to address the three identified cases.

Consider Case 1: Only One Mixed Equilibrium Case, where there is only a single mixed equilibrium. We get
\begin{equation}
\begin{aligned}
D^{A}_{12}(q_1)&=D_{1}^{A}(q_1)-D_{2}^{A}(q_1)\\
&=(r_{12}-r_{22})+L q_1.
\end{aligned}
\end{equation}

For the sake of brevity, we consider the first sub-case given by condition Eq. \ref{eq:sub-cond1}.
We have $L>0$, since $r_{11} > r_{12}$ and $r_{22} > r_{21}$. Therefore $D_{12}^{A}(q_1)$ is an increasing function of $q_1$ and
\begin{equation}
\begin{cases}
D_{12}^{A}(q_1) <0,  \hbox{if } q_1< q_{\rm opt},\\
D_{12}^{A}(q_1) =0, \hbox{if } q_1=q_{\rm opt},\\
D_{12}^{A}(q_1) > 0,  \hbox{if } q_1>q_{\rm opt}.
\end{cases}
\end{equation}
For a given $q_1$, $W_1({X})$ is quadratic in $p_1$. Also, we have:
\begin{equation}
    \label{eq:boundsW1}
    \begin{aligned}
    W_1 \left (\begin{bmatrix}
        0\\
        q_1
        \end{bmatrix}\right ) & = p_{min} D_{2}^{A}(q_1) > 0 \\
    W_1 \left (\begin{bmatrix}
        1\\
        q_1
        \end{bmatrix}\right ) & = -p_{min} D_{1}^{A}(q_1) < 0.
    \end{aligned}
\end{equation}
Since $W_1({X})$ is quadratic with a negative second derivative with respect to $p_1$, and since the inequalities in Eq. \eqref{eq:boundsW1} are strict, it admits a single root $p_1$ for $p_1 \in [0,1]$. Moreover, we have $W_1({X}) = 0$ for some $p_1$ such that:
\begin{equation}
\begin{cases}
p_1 <\frac{1}{2},  \hbox{if } q_1< q_{\rm opt},\\
p_1 =\frac{1}{2} , \hbox{if } q_1=q_{\rm opt},\\
p_1 >\frac{1}{2},  \hbox{if } q_1>q_{\rm opt}.
\end{cases}
\end{equation}
Using a similar argument, we can see that there exists a single solution for each $p_1$, and as $p_{min} \to 0$, we conclude that $W_{1}(X)=0$ whenever $p_1 \in \{0, p_{\mathrm{opt}}, 1\}$. Arguing in a similar manner we see that $W_2(X)=0$ when:

$X \in \{\begin{bmatrix} 0 \\ 0 \end{bmatrix},\begin{bmatrix} 0 \\ 1 \end{bmatrix}, \begin{bmatrix} 1 \\ 0 \end{bmatrix},\begin{bmatrix} 1 \\ 1 \end{bmatrix},\begin{bmatrix} p_{\mathrm{opt}} \\ q_{\mathrm{opt}} \end{bmatrix}\} $.

\noindent
Thus, there exists a small enough value for $p_{min}$ such that  $X^*=[p^*,q^*]\tp$ satisfies $W_2(X^*)=0$, proving Case 1).

In the proof of Case 1), we take advantage of the fact that for small enough $p_{min}$, the learning algorithm enters a stationary point, and also identified the corresponding possible values for this point. It is thus always possible to select a small enough $p_{min} >0$ such that $X^*$ approaches $X_{\mathrm{opt}}$, concluding the proof for Case 1.)

Case 2) and Case 3) can be derived in a similar manner, and the details are omitted to avoid repetition.
\end{proof}

In the next theorem, we show that the expected value of $\Delta X(t)$ has a negative definite gradient.

\begin{theorem}
\label{thm:derivative}
The matrix of partial derivatives, $\frac{\partial W(X^*)}{\partial x}$ is  negative definite.
\end{theorem}

\begin{proof}
We start the proof by writing the explicit format for $\frac{\partial W(X)}{\partial X} =
\begin{bmatrix}
  \frac{\partial W_{1}(X)}{\partial p_1} & \frac{\partial W_{1}(X)}{\partial q_1}\\
\frac{\partial W_{2}(X)}{\partial p_1} &\frac{\partial W_{2}(X)}{\partial q_1},
\end{bmatrix}$
and then computing each of the entries as below:

\begin{align*}
\frac{\partial W_{1}(X)}{\partial p_1} =& (1-2p_1) \left( D_{1}^{A}(q_1)-D_{2}^{A}(q_1)\right) - \\&p_{min} \left ( D_{1}^{A}(q_1)+ D_{2}^{A}(q_1)\right )\\
=& (1-2p_1)D^{A}_{12}(q_1) -\\&p_{min} \left ( D_{1}^{A}(q_1)+ D_{2}^{A}(q_1)\right).
\end{align*}

\begin{align*}
\frac{\partial W_{1}(X)}{\partial q_1} =&\,p_1(1-p_1)L-p_{min}(p_1(r_{11}-r_{12})+\\&(1-p_1)(r_{22}-r_{21})).
\end{align*}

\begin{align*}
\frac{\partial W_{2}(X)}{\partial p_1} =& q_1(1-q_1) L^{\prime} -p_{min}  ( q_1 (c_{11}-c_{12})  +\\& (1-q_1)(c_{22}-c_{21})).
\end{align*}

\begin{align*}
\frac{\partial W_{2}(X)}{\partial q_1} =& (1-2q_1)D^{B}_{12}(p_1) -\\&p_{min} \left ( D_{1}^{B}(p_1)+ D_{2}^{B}(p_1)\right).
\end{align*}

As seen in Theorem \ref{thm:expectedValue}, for a small enough value for $p_{min}$, we can ignore the terms that are weighted by $p_{min}$, and we will thus have $\frac{\partial W(X^*)}{\partial X} \approx  \frac{\partial W(X_{\mathrm{opt}})}{\partial X}$. We now subdivide the analysis into the three cases.

\paragraph{Case 1: No Saddle Point in pure strategies}

In this case, we have:
\begin{equation*}
D_{1}^{A}(q_{\rm opt})= D_{2}^{A}(q_{\rm opt}) \quad \text{ and }D_{1}^{B}(p_{\rm opt})= D_{2}^{B}(p_{\rm opt})
\end{equation*}
which makes
\begin{equation}
    \frac{\partial W_{1}(X_{\mathrm{opt}})}{\partial p_1} = -2p_{min} D_{1}^{A}(q_{\mathrm{opt}}).
\end{equation}
Similarly, we can compute
\begin{equation}
    \frac{\partial W_{1}(X_{\mathrm{opt}})}{\partial q_1} = (1-2p_{min})p_{\mathrm{opt}}(1-p_{\mathrm{opt}})L.
\end{equation}
The entry $\frac{\partial W_{2}(X_{\mathrm{opt}})}{\partial p_1}$ can be simplified to:
\begin{equation}
\frac{\partial W_{2}(X_{\mathrm{opt}})}{\partial p_1} = (1-2p_{min})q_{\mathrm{opt}}(1-q_{\mathrm{opt})}L^{\prime}
\end{equation}
and
\begin{equation}
\frac{\partial W_{2}(X_{\mathrm{opt}})}{\partial q_1} = 2p_{min}D_1^B(p_{\mathrm{opt}})
\end{equation}
resulting in:
\begin{equation}
\label{eq:diffW}
\resizebox{\hsize}{!}{$
    \frac{\partial W(X_{\mathrm{opt}})}{\partial X} =
    \begin{bmatrix}
      -2p_{min} D_{1}^{A}(q_{\mathrm{opt}}) & (1-2p_{min})p_{\mathrm{opt}}(1-p_{\mathrm{opt}})L\\
      (1-2p_{min})q_{\mathrm{opt}}(1-q_{\mathrm{opt})}L^{\prime} & -2p_{min}D_1^B(p_{\mathrm{opt}})
    \end{bmatrix}$}.
\end{equation}

We know that this case can be divided into two sub-cases. Let us consider the first sub-case given by:
\begin{equation}
    r_{11} > r_{21},  r_{12} < r_{22}; c_{11} < c_{12},  c_{21} > c_{22},\label{eq:L>0,Lprime<0}
\end{equation}

Thus, $L>0$ and $L^{\prime}<0$ as a consequence of Eq. \eqref{eq:L>0,Lprime<0}

Thus, the matrix given in Eq. \eqref{eq:diffW} satisfies:
\begin{equation}
    \mathrm{det}\left(\frac{\partial W({X}_{\mathrm{opt}})}{\partial x}\right) > 0 \, \text{ , } \,
    \mathrm{trace}\left(\frac{\partial W({X}_{\mathrm{opt}})}{\partial x}\right) < 0,
\end{equation}
which implies the $2\times 2$ matrix is negative definite.

\paragraph{Case 2: Only one single pure equilibrium}
In Case 2, corresponds to. According to this case:
 $(r_{11}-r_{21})(r_{12}-r_{22})>0$ or
$(c_{11}-c_{12})(c_{21}-c_{22})>0$

The condition for only one pure equilibrium can be divided into four different sub-cases.

Without loss of generality, we can consider a particular sub-case  where $q_{\rm opt}=1$ and $p_{\rm opt}=1$. This reduces to $r_{11}-r_{21}>0$ and $c_{11}-c_{12}>0$.

Computing the entries of the matrix for this case yields:
\begin{equation}
    \frac{\partial W_{1}(X_{\mathrm{opt}})}{\partial p_1} = -(r_{11}-r_{21})-p_{min} (r_{11}+r_{21}),
\end{equation}
and
\begin{equation}
    \frac{\partial W_{1}(X_{\mathrm{opt}})}{\partial q_1} = -p_{min}(r_{11}-r_{12}).
\end{equation}
The entry $\frac{\partial W_{2}(X_{\mathrm{opt}})}{\partial p_1}$ can be simplified to:
\begin{equation}
\frac{\partial W_{2}(X_{\mathrm{opt}})}{\partial p_1} = -p_{min}(c_{11}-c_{12})
\end{equation}
and
\begin{equation}
\frac{\partial W_{2}(X_{\mathrm{opt}})}{\partial q_1} = -(c_{11}-c_{12})-p_{min} (c_{11}+c_{12})
\end{equation}
resulting in:
\begin{align}
   & \frac{\partial W(X_{\mathrm{opt}})}{\partial X} \\\notag =
    & \scriptsize{\begin{bmatrix}
      -(r_{11}-r_{21})-p_{min} (r_{11}+r_{21}) & -p_{min}(r_{11}-r_{12})\\
     - p_{min}(c_{11}-c_{12}) & -(c_{11}-c_{12})-p_{min} (c_{11}+c_{12})
    \end{bmatrix}}.
\end{align}
The matrix in \eqref{eq:diffWCase2} satisfies:
\begin{equation}
    \mathrm{det}\left(\frac{\partial W(X_{\mathrm{opt}})}{\partial X}\right) > 0 \, \text{ , } \,
    \mathrm{trace}\left(\frac{\partial W(X_{\mathrm{opt}})}{\partial X}\right) < 0
\end{equation}
for a sufficiently small value of $p_{min}$, which again implies that the $2\times 2$ matrix is negative definite.

\paragraph{Case 3: Two pure equlibria and one mixed equilibrium}

In this case, $(r_{11}-r_{21})(r_{12}-r_{22})<0$, 
$(c_{11}-c_{12})(c_{21}-c_{22})<0$ and 
$(r_{11}-r_{21})(c_{11}-c_{12})>0$.

Without loss of generality, we suppose that  $(p_{\rm opt},q_{\rm opt})=(1,1)$ and  $(p_{\rm opt},q_{\rm opt})=(0,0)$ are the two pure Nash equilibria.
This corresponds to
a sub-case where:
\begin{equation}
  r_{11}-r_{21}>0  c_{11}-c_{12}>0,  r_{22}-r_{12}>0  c_{22}-c_{21}>0,
    \label{eq:case3_sub-cond}
\end{equation}

$r_{11}-r_{21}>0$ and $c_{11}-c_{12}>0$ because of the Nash equilibrium $(p_{\rm opt},q_{\rm opt})=(1,1)$
Similarly, 
$r_{22}-r_{12}>0$ and $c_{22}-c_{21}>0$
because of the Nash equilibrium $(p_{\rm opt},q_{\rm opt})=(1,1)$

Whenever  $(p_{\rm opt},q_{\rm opt})=(1,1)$, we obtain stability of the fixed point as demonstrated in the previous case, case 2.

Now, let us consider the stability for $(p_{\rm opt},q_{\rm opt})=(0,0)$.

Computing the entries of the matrix for this case yields:
\begin{equation}
    \frac{\partial W_{1}(X_{\mathrm{opt}})}{\partial p_1} = (r_{12}-r_{22})-p_{min} (r_{12}+r_{22}),
\end{equation}
and
\begin{equation}
    \frac{\partial W_{1}(X_{\mathrm{opt}})}{\partial q_1} = -p_{min}(r_{22}-r_{21}).
\end{equation}
The entry $\frac{\partial W_{2}(X_{\mathrm{opt}})}{\partial p_1}$ can be simplified to:
\begin{equation}
\frac{\partial W_{2}(X_{\mathrm{opt}})}{\partial p_1} = -p_{min}(c_{22}-c_{12})
\end{equation}
and
\begin{equation}
\frac{\partial W_{2}(X_{\mathrm{opt}})}{\partial q_1} = (c_{21}-c_{22})-p_{min} (c_{21}+c_{22})
\end{equation}
resulting in:
\begin{align}
\label{eq:diffWCase2}
   & \frac{\partial W(X_{\mathrm{opt}})}{\partial X} \\\notag =
    & \scriptsize{\begin{bmatrix}
      (r_{12}-r_{22})-p_{min} (r_{12}+r_{22}) & -p_{min}(r_{22}-r_{21})\\
     -p_{min}(c_{22}-c_{12}) & (c_{21}-c_{22})-p_{min} (c_{21}+c_{22})
    \end{bmatrix}}.
\end{align}
The matrix in \eqref{eq:diffWCase2} satisfies:
\begin{equation}
    \mathrm{det}\left(\frac{\partial W(X_{\mathrm{opt}})}{\partial X}\right) > 0 \, \text{ , } \,
    \mathrm{trace}\left(\frac{\partial W(X_{\mathrm{opt}})}{\partial X}\right) < 0
\end{equation}
for a sufficiently small value of $p_{min}$, which again implies that the $2\times 2$ matrix is negative definite.

Now, what remains to be shown is that the mixed Nash equilibrium in this case is unstable.

\begin{equation}
\resizebox{\hsize}{!}{$
    \frac{\partial W(X_{\mathrm{opt}})}{\partial X} =
    \begin{bmatrix}
      -2p_{min} D_{1}^{A}(q_{\mathrm{opt}}) & (1-2p_{min})p_{\mathrm{opt}}(1-p_{\mathrm{opt}})L\\
      (1-2p_{min})q_{\mathrm{opt}}(1-q_{\mathrm{opt})}L^{\prime} & -2p_{min}D_1^B(p_{\mathrm{opt}})
    \end{bmatrix}$}.
\end{equation}

Using Eq. \ref{eq:case3_sub-cond}, we can see that
$L>0$ and $L^{\prime}>0$ and thus:

\begin{equation}
\mathrm{det}\left(\frac{\partial W(X_{\mathrm{opt}})}{\partial X}\right) < 0
\end{equation}

\end{proof}

\begin{theorem}
We consider the update equations given by the $L_{R-I}$ scheme.
For a sufficiently small $p_{min}$ approaching $0$,  and as $\theta \to 0$ and  as time goes to infinity:

$\begin{bmatrix} E(p_1(t)) & E(q_1(t)) \end{bmatrix}$  $\rightarrow$ $\begin{bmatrix}p_{\mathrm{opt}}^* &q_{\mathrm{opt}}^*\end{bmatrix}$

where $\begin{bmatrix}p_{\mathrm{opt}}^* &q_{\mathrm{opt}}^*\end{bmatrix}$ corresponds to a Nash equilibrium of the game.
\end{theorem}

\begin{proof}

The proof of the result is obtained by virtue of applying a classical result due 
 due to to Norman \cite{Norman1972}, given in the Appendix \ref{sec:Appendix}, in the interest of completeness.

Norman theorem has been traditionally used to prove considerable amount of the results in the field of LA. 
In the context of game theoretical LA schemes, Norman theorem has been adapted by Lakshmivarahan and Narendra to derive similar convergence properties of the $L_{R-\epsilon P}$ \cite{lakshmivarahan1982learning} for the zero-sum game. It is straightforward  to verify that Assumptions (1)-(6) as required for Norman's result in the appendix are satisfied.
Thus, by further invoking Theorem \ref{thm:expectedValue} and Theorem \ref{thm:derivative} , the result follows.
\end{proof}

\section{Game Theoretical LA Algorithm based on the $S$-Learning with Artificial Barriers}
\label{sec:scheme_S_learning}

In this section, we give the update equations for the LA when the environment is of $S-$ type.

In the case of $S-$ type, the game is defined by two payoff matrices, $R$ and $C$ describing a deterministic feedback of player $A$ and player $B$ respectively.

All the entries of both matrices are deterministic numbers like in classical game theory settings.

 The environment returns $u_i^A(t)$: the payoff defined by the matrix $R$ for  player  $A$    at time $t$ whenever player $A$  question chooses an  action $i \in \{1,2\}$.

 The update rule for the  player $A$  that takes into account $u_i^A(t)$ is given by:
\begin{equation}
\label{eq:algorithmGain_S_Learnign}
\begin{aligned}
p_i(t+1) & =   p_i(t)+\theta u_i^A (p_{max}-p_i(t))\\
p_s(t+1) & =    p_s(t)+\theta u_i^A (p_{min}-p_s(t)) &\textrm{for } \quad s\ne i.
\end{aligned}
\end{equation}
\noindent
where $\theta$ is a learning parameter.

Note $u_i^A$ is the feedback for action $i$ of the  player $A$ which is one entry in the $i^th$ row of the matrix $R$, depending on the action of the player $B$.

Similarly we can define  $u_i^B(t)$ the payoff defined by the matrix $C$ for  player  $B$    at time $t$ whenever player $B$  question chooses an  action $i \in \{1,2\}$.
 
 For instance , if at time $t$, player $A$ takes action $1$ and player $B$ takes action $2$, then $u_1^A(t)=r_{12}$  and  $u_2^B(t)=c_{21}$.

The update rules for player $B$ can be obtained by analogy to those given for player $A$.

\begin{theorem}
We consider the update equations given by the $S-$ Learning  scheme given above in this Section.
For a sufficiently small $p_{min}$ approaching $0$,  and as $\theta \to 0$ and  as time goes to infinity:

$\begin{bmatrix} E(p_1(t)) & E(q_1(t)) \end{bmatrix}$  $\rightarrow$ $\begin{bmatrix}p_{\mathrm{opt}}^* &q_{\mathrm{opt}}^*\end{bmatrix}$

where $\begin{bmatrix}p_{\mathrm{opt}}^* &q_{\mathrm{opt}}^*\end{bmatrix}$ corresponds to a Nash equilibrium of the game.
\end{theorem}

\begin{proof}
The proofs of this theorem follows the same lines as the proofs give in Section \ref{sec:scheme_L_R_I} and are omitted here for the sake of brevity.
\end{proof}

\section{Experimental results}
\label{sec:simualations_L_RI}

In this Section, we focus on providing thorough experiments for $L_{R-I}$ scheme. Some experiments of $S-$ LA for handling   $S-$ type environments are given in the Appendix \ref{sec:appendix_B} that mainly aim to verify our theoretical findings.

To verify the theoretical properties of the proposed learning algorithm, we conducted several simulations that will be presented in this section.  By using different instances of the payoff matrices $R$ and $C$, we can experimentally cover the three cases referred to in Section \ref{sec:scheme_L_R_I}. 

\subsection{Convergence in Case 1}
We examine a case of the game where only one mixed Nash equilibrium exists meaning that  there is no Saddle Point in pure strategies. The game matrices $R$ and $C$ are given by:

\begin{equation}
R=\begin{pmatrix}0.2 & 0.6 \\ 0.4 & 0.5\end{pmatrix},
\label{eq:R1}
\end{equation}
\noindent

\begin{equation}
C=\begin{pmatrix} 0.4 & 0.25 \\ 0.3 & 0.6\end{pmatrix},
\label{eq:C1}
\end{equation}
\noindent

which admits $p_{opt}=0.6667$ and $q_{opt}=0.3333$.

We ran our simulation for $5 \times 10^6$ iterations, and present the error in Table \ref{table:error_Mixed} for different values of
$p_{max}$ and $\theta$ as the difference between $X_{\mathrm{opt}}$ and the mean over time of $X(t)$ after convergence\footnote{The mean is taken over the last $10\%$ of the total number of iterations.}. The high value for the number of iterations was chosen in order to eliminate the Monte Carlo error.
A significant observation is that the error monotonically decreases as $p_{max}$ goes towards 1 (i.e., when $p_{min} \to 0$). For instance, for $p_{max}=0.998$ and $\theta=0.001$, the proposed scheme yields an error of $5.27 \times 10^{-3}$, and further reducing $\theta=0.0001$ leads to an error of $3.34 \times 10^{-3}$.

\begin{table}[htp!]
  \caption{{Error for different values of $\theta$ and $p_{max}$, when $p_{opt}=0.6667$ and $q_{opt}=0.3333$ for the game specified by the $R$ matrix given by Eq. (\ref{eq:R1}) and the $C$ matrix given by Eq. (\ref{eq:C1}).}}
\[
\begin{array}{ | c | c | c | }
\hline
	p_{max} &  \theta=0.001 &  \theta=0.0001 \\ \hline
	0.990 & 1.77 \times 10^{-2} & 2.03 \times 10^{-2} \\ \hline
	0.991 & 1.71 \times 10^{-2} & 1.69 \times 10^{-2} \\ \hline
	0.992 & 1.33 \times 10^{-2} & 1.54 \times 10^{-2} \\ \hline
	0.993 & 1.32 \times 10^{-2} & 1.52 \times 10^{-2} \\ \hline
	0.994 & 1.18 \times 10^{-2} & 1.02 \times 10^{-2} \\ \hline
	0.995 & 1.17 \times 10^{-2} & 7.86 \times 10^{-3} \\ \hline
	0.996 & 8.50 \times 10^{-3} & 6.37 \times 10^{-3} \\ \hline
	0.997 & 5.57 \times 10^{-3} & 4.43 \times 10^{-3} \\ \hline
	0.998 & 5.27 \times 10^{-3} & 3.34 \times 10^{-3} \\ \hline
\end{array}
\]
 \label{table:error_Mixed}
\end{table}

\begin{figure}[htp!]
\centering
\includegraphics[width=8cm]{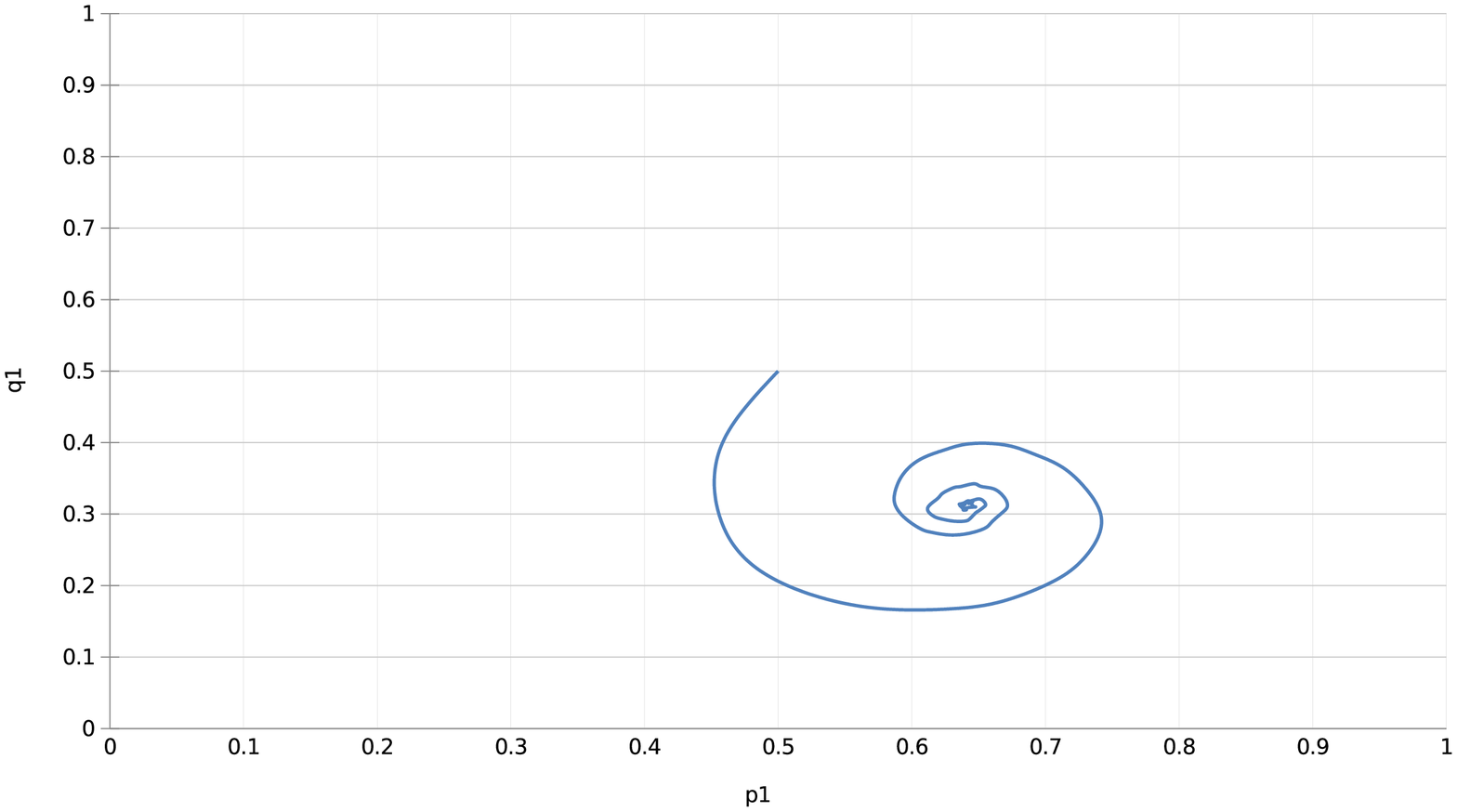}
\caption{Trajectory of  $[p_1(t), q_1(t)]\tp$ for the case of the $R$ matrix given by Eq. (\ref{eq:R1}) and the $C$ matrix given by Eq. (\ref{eq:C1}) with $p_{opt}=0.6667$ and $q_{opt}=0.3333$, and using $p_{max}=0.99$ and $\theta=0.01$.}
\label{fig:ensembleTrajectory_Mixed}
\end{figure}

The behavior scheme is illustrated in Figure \ref{fig:ensembleTrajectory_Mixed} showing the trajectory of the mixed strategies for both players (given by $X(t)$) for an ensemble of 1,000 runs using $\theta=0.01$ and $p_{max}=0.99$.

The trajectory of the ensemble enables us to notice the mean evolution of the mixed strategies. The spiral pattern results from one of the players adjusting to the strategy used by the other before the former learns by readjusting its strategy. The process is repeated, thus leading to more minor corrections until the players reach the Nash equilibrium.

The process can be visualized in Figure \ref{fig:trajectory_Mixed} presenting the time evolution  of the strategies of both players for a single experiment with $p_{max}=0.99$ and $\theta=0.00001$ over $3\times10^7$ steps. We observe an oscillatory behavior which vanishes as the players play for more iterations. It is worth noting that a larger value of $\theta$ will cause more steady state error (as specified in Theorem \ref{thm:expectedValue}), but it will also disrupt this behavior as the players take larger updates whenever they receive a reward. Furthermore,  decreasing $\theta$ results in a smaller convergence error, but also affects negatively the convergence speed as more iterations are required to achieve convergence.
Figure \ref{fig:trajectory_Mixed_smaller_theta} depicts the trajectories of the probabilities $p_1$ and $q_1$ for the same settings as those in Figure \ref{fig:trajectory_Mixed}.

\begin{figure}[htp!]
\centering
\includegraphics[width=8cm]{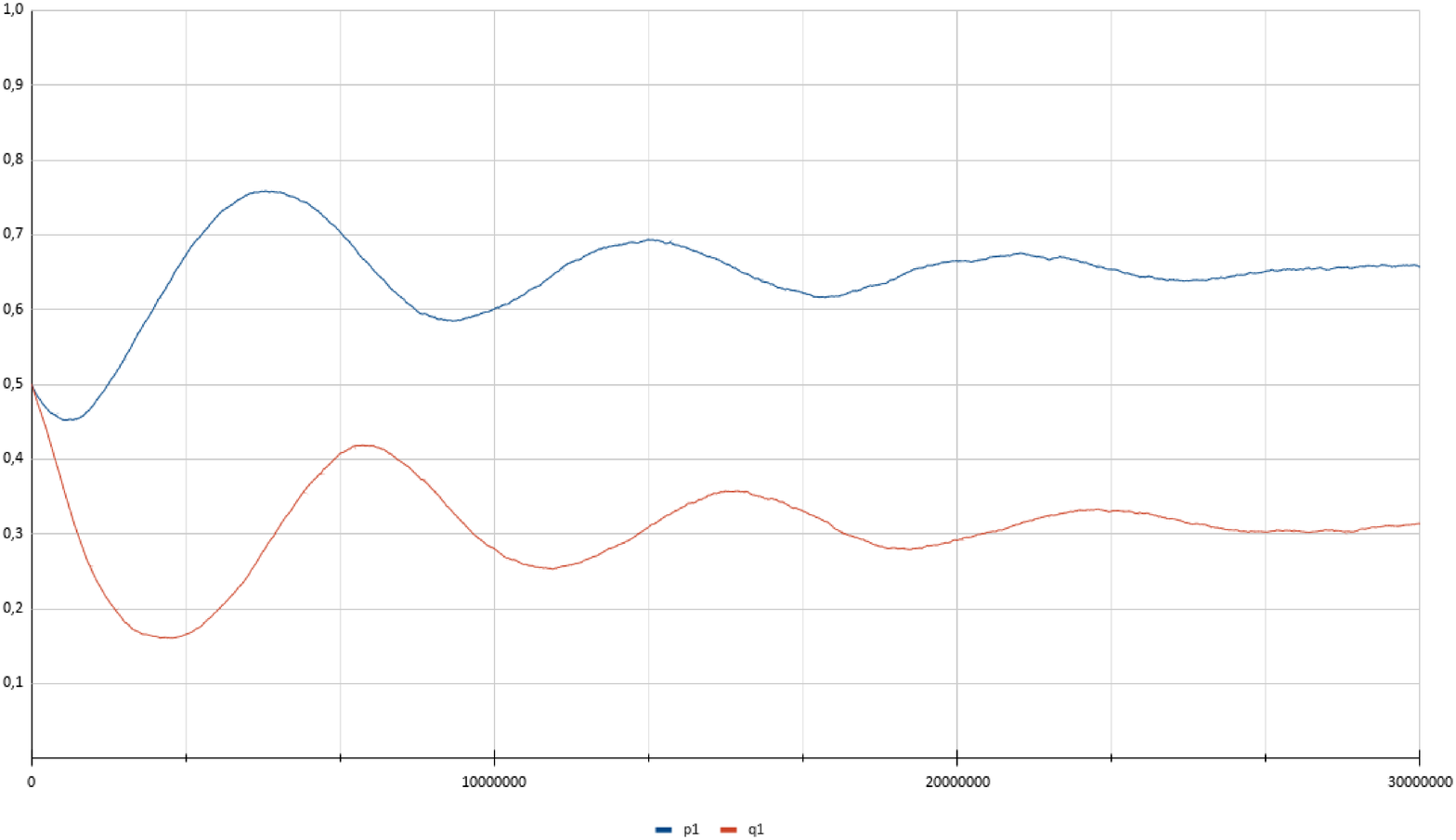}
\caption{Time Evolution  $X(t)$ for the case of the $R$ matrix given by Eq. (\ref{eq:R1}) and the $C$ matrix given by Eq. (\ref{eq:C1}) with $p_{opt}=0.6667$ and $q_{opt}=0.3333$, and using $p_{max}=0.99$ and $\theta=0.00001$.}
\label{fig:trajectory_Mixed}
\end{figure}

\begin{figure}[htp!]
\centering
\includegraphics[width=8cm]{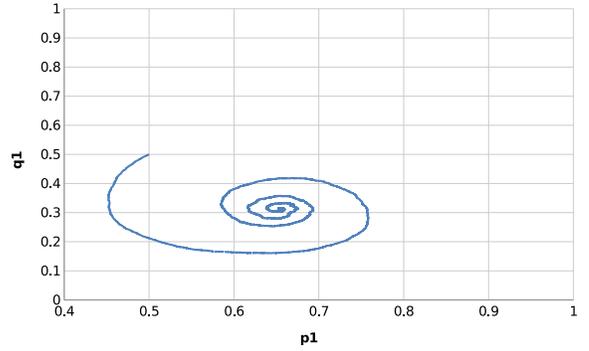}
\caption{Trajectory of $X(t)$ where
$p_{opt}=0.6667$ and $q_{opt}=0.3333$, using $p_{max}=0.99$ and $\theta=0.00001$.}
\label{fig:trajectory_Mixed_smaller_theta}
\end{figure}


Now, we turn our attention to the analysis of the the deterministic Ordinary Differential Equation (ODE) corresponding to our LA with barriers and plot it in Figure \ref{fig:ODE_Mixed}.
The trajectory of the ODE is conform with our intuition and the results of the LA  run in Figure \ref{fig:trajectory_Mixed_smaller_theta}.
The two ODE are given by:
\begin{equation}
\resizebox{0.9\hsize}{!}{$
\begin{aligned}
   \frac{dp_1}{dt}= W_1({X})
    = \ \ &  p_1 (p_{max} - p_1) D_1^A(q_1) + (1-p_1) (p_{min} - p_1) D_2^A(q_1),  \\
\end{aligned}
$}
\end{equation}

and,

\begin{equation}
\resizebox{0.9\hsize}{!}{$
\begin{aligned}
   \frac{dq_1}{dt}= W_1({X})
    = \ \ &  p_1 (p_{max} - p_1) D_1^A(q_1) + (1-p_1) (p_{min} - p_1) D_2^A(q_1),  \\
\end{aligned}
$}
\end{equation}

To obtain the ODE for a particular example, we need just to replace the entries of $R$ and $C$ in the ODE by their values. In this sense to plot the ODE trajectories we only need to know $R$ and $C$ and of course $p_{max}$.

\begin{figure*}[htp!]
\centering
\includegraphics[scale=0.4]{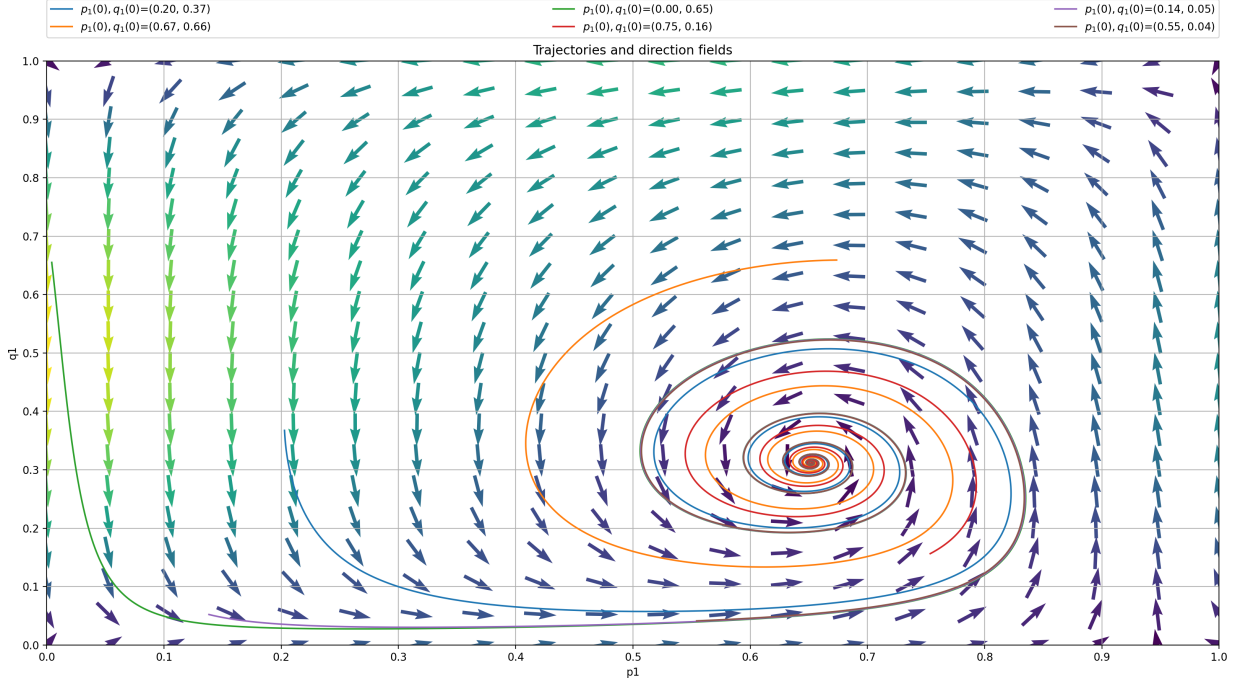}
\caption{Trajectory of  ODE  using $p_{max}=0.99$ for case 1.}
\label{fig:ODE_Mixed}
\end{figure*}

\subsection{Case 2: One Pure equilibrium}

At this juncture, we shall experimentally show that the scheme possess still plausible convergence properties even in case where there is a single saddle point in pure strategies and that our proposed LA will approach the optimal pure equilibria.
For this sake, we consider a case of the game where there is a single pure equilibrium which falls in the category of Case 2 with $p_{opt}=1$ and $q_{opt}=0$. The payoff matrices $R$ and $C$ for the games are given by:

\begin{equation}
R=\begin{pmatrix}0.7 & 0.9 \\ 0.6 & 0.8\end{pmatrix},
\label{eq:R2}
\end{equation}
\noindent

\begin{equation}
C=\begin{pmatrix} 0.6 & 0.8 \\ 0.8 & 0.9 \end{pmatrix},
\label{eq:C2}
\end{equation}
\noindent

We first show the convergence errors of our method  in Table \ref{tab:gameD1}. As in the previous simulation for Case 1, the errors are on the order to $10^{-2}$ for larger values of $p_{max}$. We also observe that steady state error is slightly higher compared to the previous case of mixed Nash described by Eq. \eqref{eq:R1} and Eq. \eqref{eq:C1} which is treated in the previous section. 

 \begin{table}[htp!]
  \caption{Error for different values of $\theta$ and $p_{max}$ for the game specified by the R matrix and the C matrix given by Eq. \eqref{eq:R2} and Eq. \eqref{eq:C2}.}
\[
\begin{array}{ | c | c | c | }
\hline
	p_{max} &  \theta=0.0001 &  \theta=0.00001 \\ \hline
	0.990 & 6.57 \times 10^{-2} & 6.51\times 10^{-2} \\ \hline
	0.991 & 5.88 \times 10^{-2} & 5.82 \times 10^{-2} \\ \hline
	0.992 & 5.30 \times 10^{-2} & 5.21 \times 10^{-2} \\ \hline
	0.993 & 4.67 \times 10^{-2} & 4.64 \times 10^{-2} \\ \hline
	0.994 & 4.00 \times 10^{-2} & 4.02 \times 10^{-2} \\ \hline
	0.995 & 3.36 \times 10^{-2} & 3.38 \times 10^{-2} \\ \hline
	0.996 & 2.68 \times 10^{-2} & 2.64 \times 10^{-2} \\ \hline
	0.997 & 2.04 \times 10^{-2} & 2.08 \times 10^{-2} \\ \hline
	0.998 & 1.40 \times 10^{-2} & 1.37 \times 10^{-2} \\ \hline
\end{array}
\]
\label{tab:gameD1}
\end{table}

We then plot the ODE for $p_{max}=0.99$ as shown in Figure \ref{fig:ODE_one_pure_099}. According to the ODE in Figure \ref{fig:ODE_one_pure_099}, we are expecting that the LA will converge towards the attractor of the ODE which corresponds to $(p^*,q^*)=0.917, 0.040)$ as $\theta$ goes to zero.
We see that $(p^*,q^*)=(0.917, 0.040)$ approaches $(p_{opt},q_{opt})=(1, 0)$ but there is still a gap between them.
This is also illustrated in
Figure \ref{fig:case2_te_pmax099}  where we also consistently observe that the LA converges towards $(p^*,q^*)=(0.916, 0.041)$ after running our LA for 30,000 iterations with an ensemble of 1,000 experiments.


\begin{figure*}[htp!]
\centering
\includegraphics[scale=0.4]{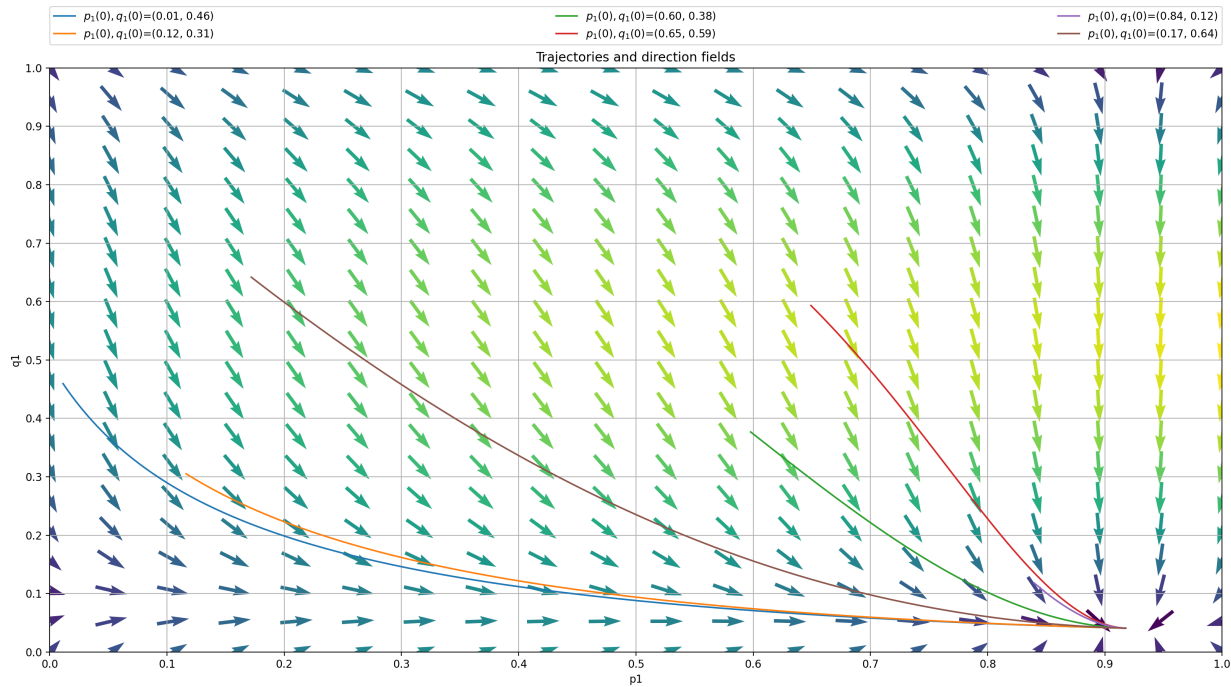}
\caption{Trajectory of  the deterministic ODE  using $p_{max}=0.99$ for case 2.}
\label{fig:ODE_one_pure_099}
\end{figure*}

\begin{figure}[htp!]
\centering
\includegraphics[width=8cm]{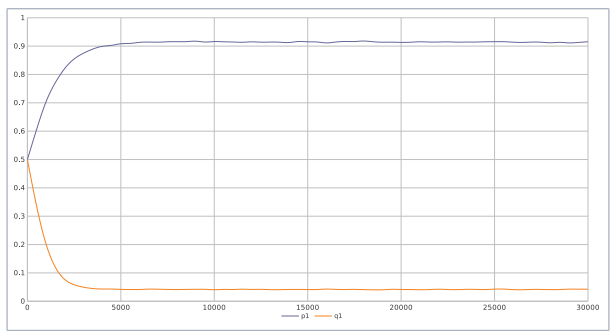}
\caption{Time evolution over time of $X(t)$ for $\theta=0.01$ and $p_{max}=0.99$ for the case of the $R$ matrix given Eq. \eqref{eq:R2} and for the $C$ matrix given by Eq. \eqref{eq:C2}.}
\label{fig:case2_te_pmax099}
\end{figure}

Observing the small dispersancy between 
between $(p^*,q^*)=(0.917, 0.040)$ and $(p_{opt},q_{opt})=(1, 0)$ from the ODE and from the LA trajectory as shown in Figure  \ref{fig:ODE_one_pure_099} and Figure 
\ref{fig:case2_te_pmax099} motivates us to choose even a larger value of $p_{max}$. Thus, we increase $p_{max}$ from $0.99$ to 
$0.999$ and observe the expected convergence results from the ODE in 
in Figure \ref{fig:ODE_one_pure_0.999}. We observe a single attraction point close  of the ODE close to the pure Nash equilibrium. We can read from the ODE trajectory that  $(p^*,q^*)=(0.991, 0.004)$  which is closer $(p_{opt},q_{opt})=(1, 0)$ than the previous case with a smaller $p_{max}$.

\begin{figure*}[htp!]
\centering
\includegraphics[scale=0.4]{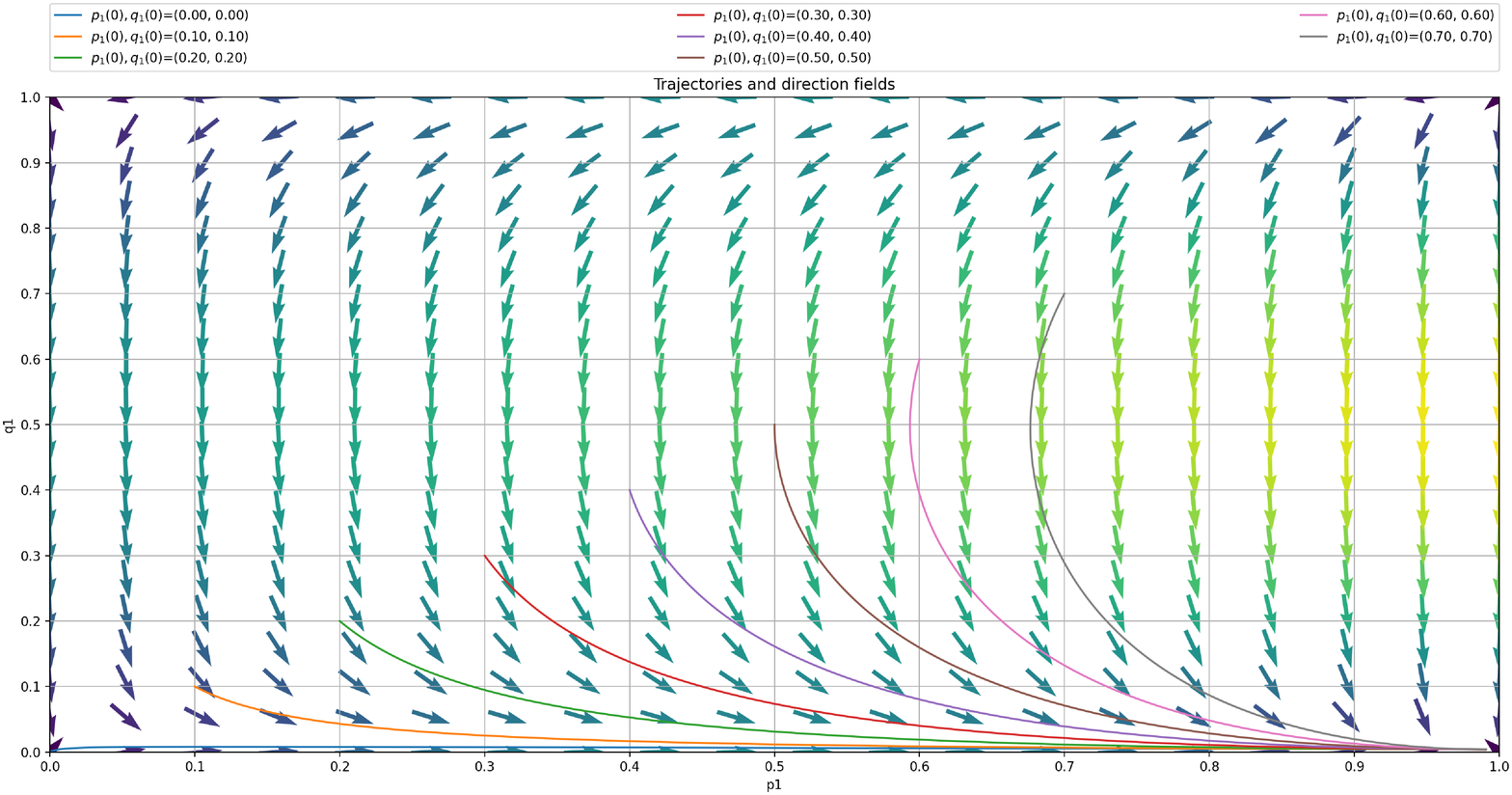}
\caption{Trajectory of  ODE  using $p_{max}=0.999$ for case 2.}
\label{fig:ODE_one_pure_0.999}
\end{figure*}

In Figure \ref{fig:D2a}, we depict the time evolution of the two components of the vector $X(t)$ using the proposed algorithm for an ensemble of 1,000 runs. In the case of having a Pure Nash equilibrium, there is no oscillatory behavior as when a player assigns more probability to an action, since the other player reinforces the strategy. However, Figure \ref{fig:D2a} could mislead the reader to believe that the LA method has converged to a pure strategy for both players. In order to clarify that we are not converging to an absorbing state for the player A, we provide Figure \ref{fig:D2b} which zooms on Figure \ref{fig:D2a} around the region where the strategy of player A has converged in order to visualize that its maximum first action probability is limited by $p_{max}$, as per the design of our updating rule.
Similarly, we zoom on the evolution of the first action probability of player B in Figure \ref{fig:D2c}. We observe that the first action instead to converging to zero as it would be if we did not have absorbing barriers, its rather converges to a small probability limited by $p_{min}$ which approaches zero. 
Such propriety of evading lock in probability even for pure optimal strategies and which emanates from the ergodicity of our  $L_{R-I}$ scheme with absorbing barriers  is a desirable property specially 
when the payoff matrices are time-varying and thus the optimal equilibrium point might change over time. Such a case deserves a separate study to better understand the behavior of the scheme and to also understand the effect of the tuning parameters and how to control and vary them in this case to yield a compromise between learning and forgetting stale information.

\begin{figure}[ht]
\begin{subfigure}{.5\textwidth}
  \centering
\includegraphics[scale=0.3]{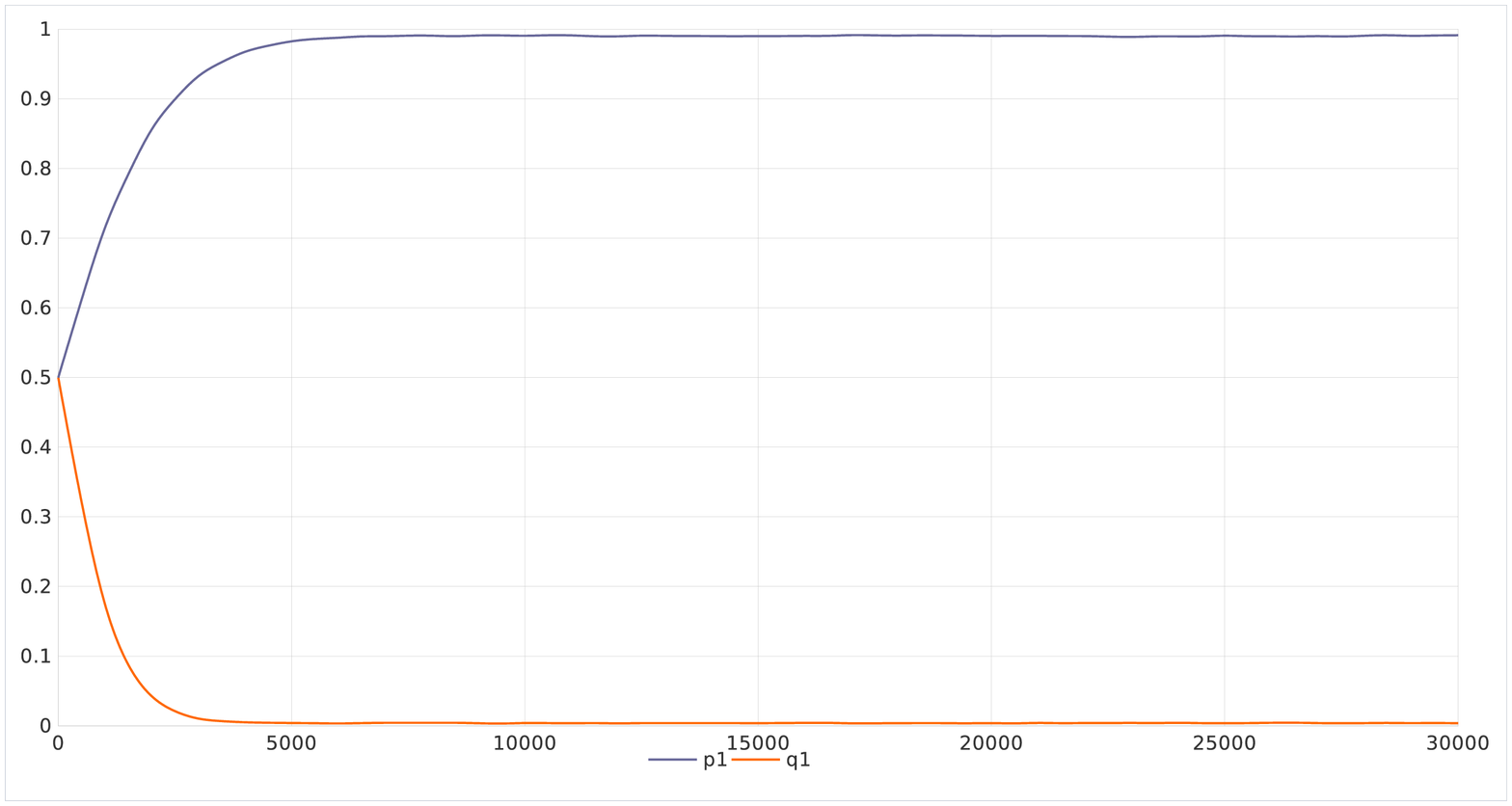}
  \caption{Evolution over time of $X(t)$.}
  \label{fig:D2a}
\end{subfigure}
\begin{subfigure}{.5\textwidth}
  \centering
\includegraphics[scale=0.3]{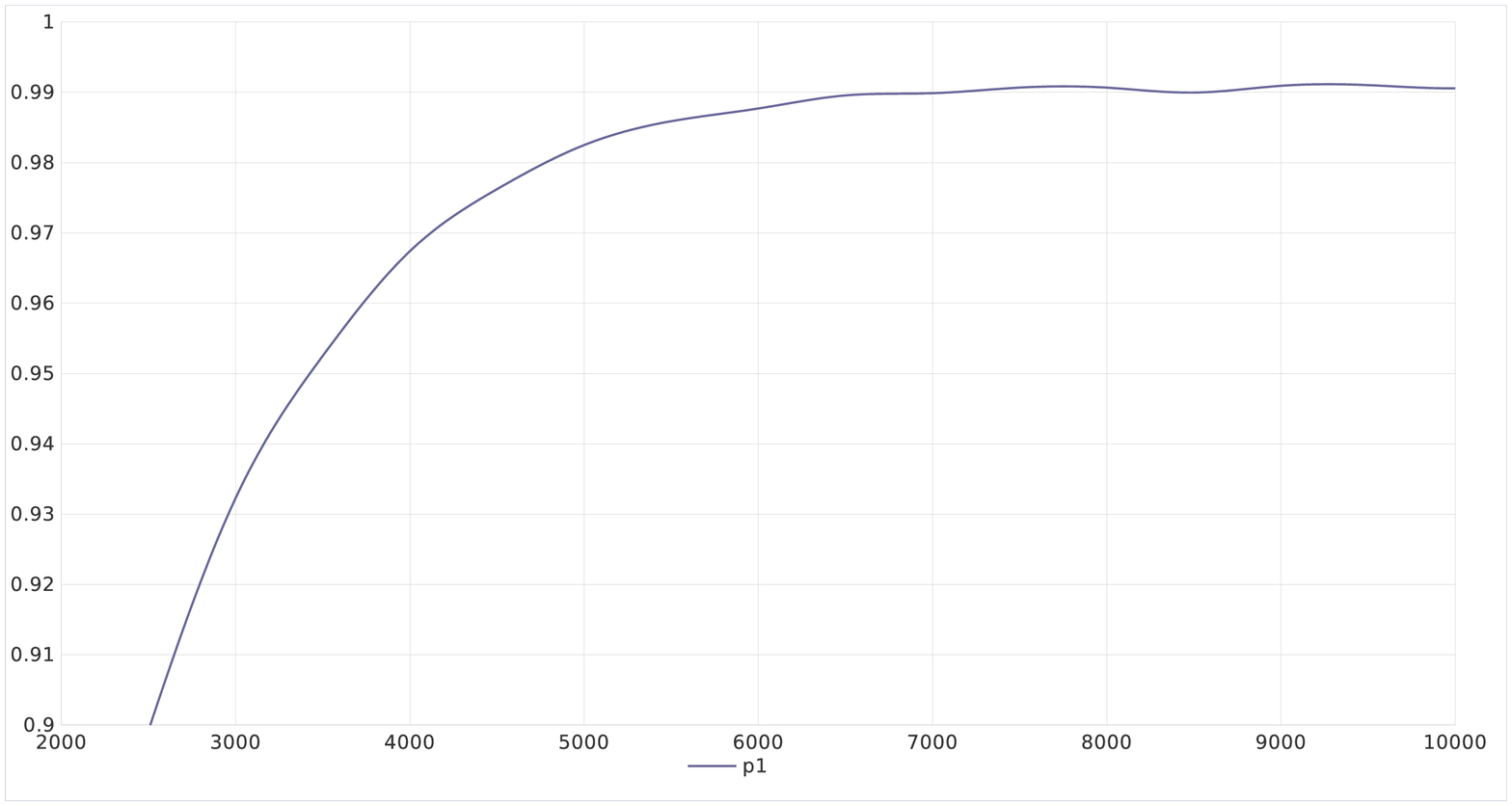}
  \caption{Zoomed version for player A strategy.}
  \label{fig:D2b}
\end{subfigure}

\begin{subfigure}{.5\textwidth}
  \centering
\includegraphics[scale=0.3]{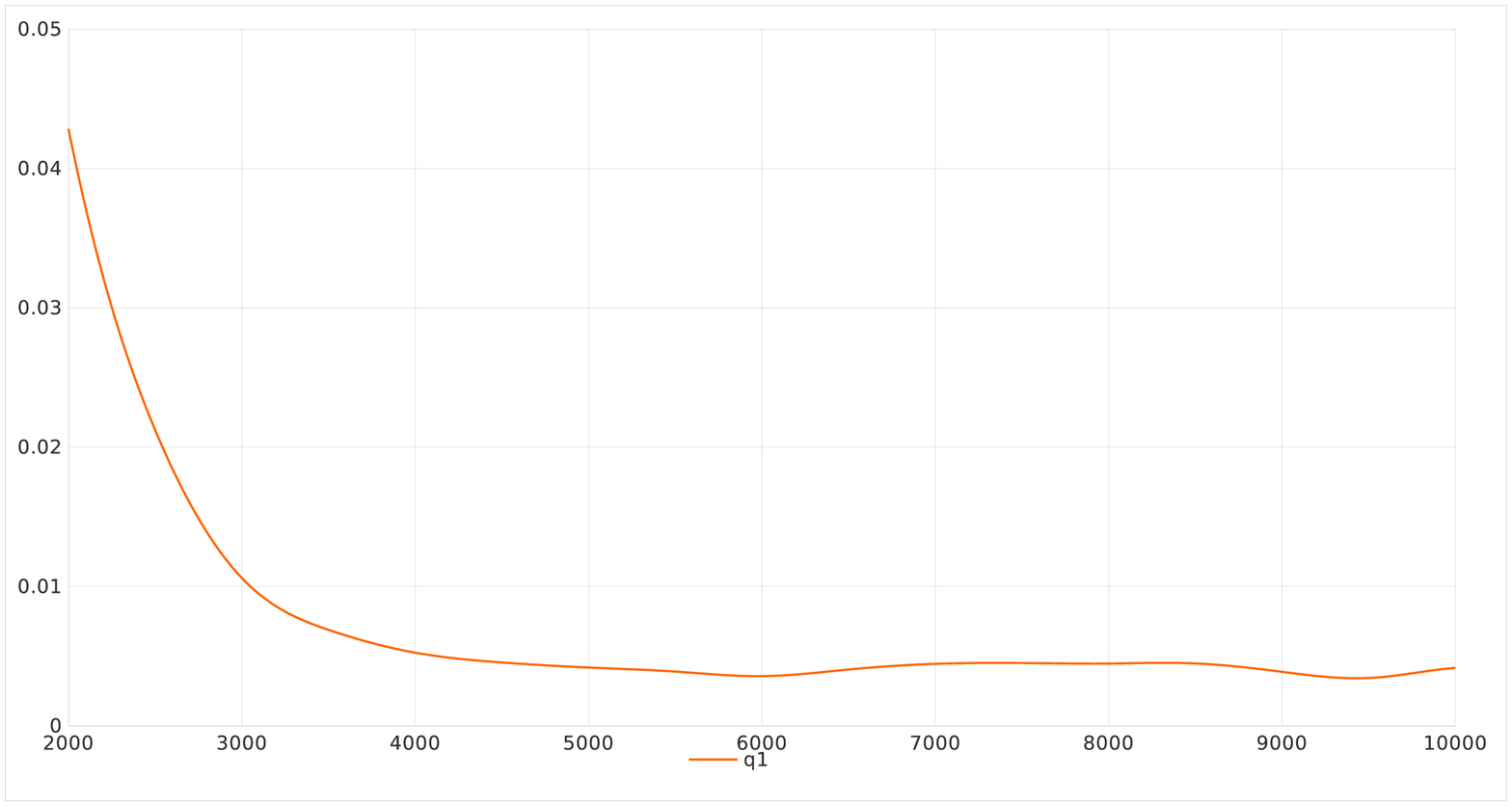}
  \caption{Zoomed version for player B strategy.}
  \label{fig:D2c}
\end{subfigure}
\caption{The figure shows a) the evolution over time of $X(t)$ for $\theta=0.01$ and $p_{max}=0.999$ when applied to game with payoffs specified by the $R$ matrix and the $C$ matrix given by Eq. \eqref{eq:R2} and Eq. \eqref{eq:C2}., and b) is a zoomed version around player A strategy c) and is a zoomed version around  player B strategy.}
\label{fig:timeTrajectoryD2}
\end{figure}

Figure \ref{fig:timeTrajectoryD2} depicts the time evolution of the probabilities for each player,  with $\theta=0.01$, $p_{max}=0.999$ and for an ensemble with 1,000 runs.

\subsection{Case 3: 2 Pure equilibria and 1 mixed}

Now, we shall consider the last case 3.

As an instance of case 3, we consider the payoff matrices $R$ and $C$ given by:

\begin{equation}
R=\begin{pmatrix}0.3 & 0.1 \\ 0.2 & 0.3\end{pmatrix},
\label{eq:R3}
\end{equation}
\noindent

\begin{equation}
C=\begin{pmatrix} 0.3 & 0.2 \\ 0.1 & 0.2 \end{pmatrix},
\label{eq:C3}
\end{equation}

In Figure \ref{fig:Case3_trajectories_LA}, we plot 9 trajectories for the LA for a number of iterations is 1,000,000.
We observe that depending on the initial conditions, our LA  converges to one of the two pure equilibria which is usually the closest to the starting point. We have also performed extensive simulations with initial values $(0.5, 0.5)$ of the probabilities and we found that almost $50\%$ of the time the LA converges to the Nash equilibrium close to $(1,1)$ and 
$50\%$ close to $(0,0)$.
As a future work, we would like to explore how to push the LA to favor one of the two equilibria as there is usually an equilibrium that is superior to the other, and thus it is more desirable for both players to converge to the superior Nash equilibrium.

\begin{figure}[htp!]
\centering
\includegraphics[width=9cm]{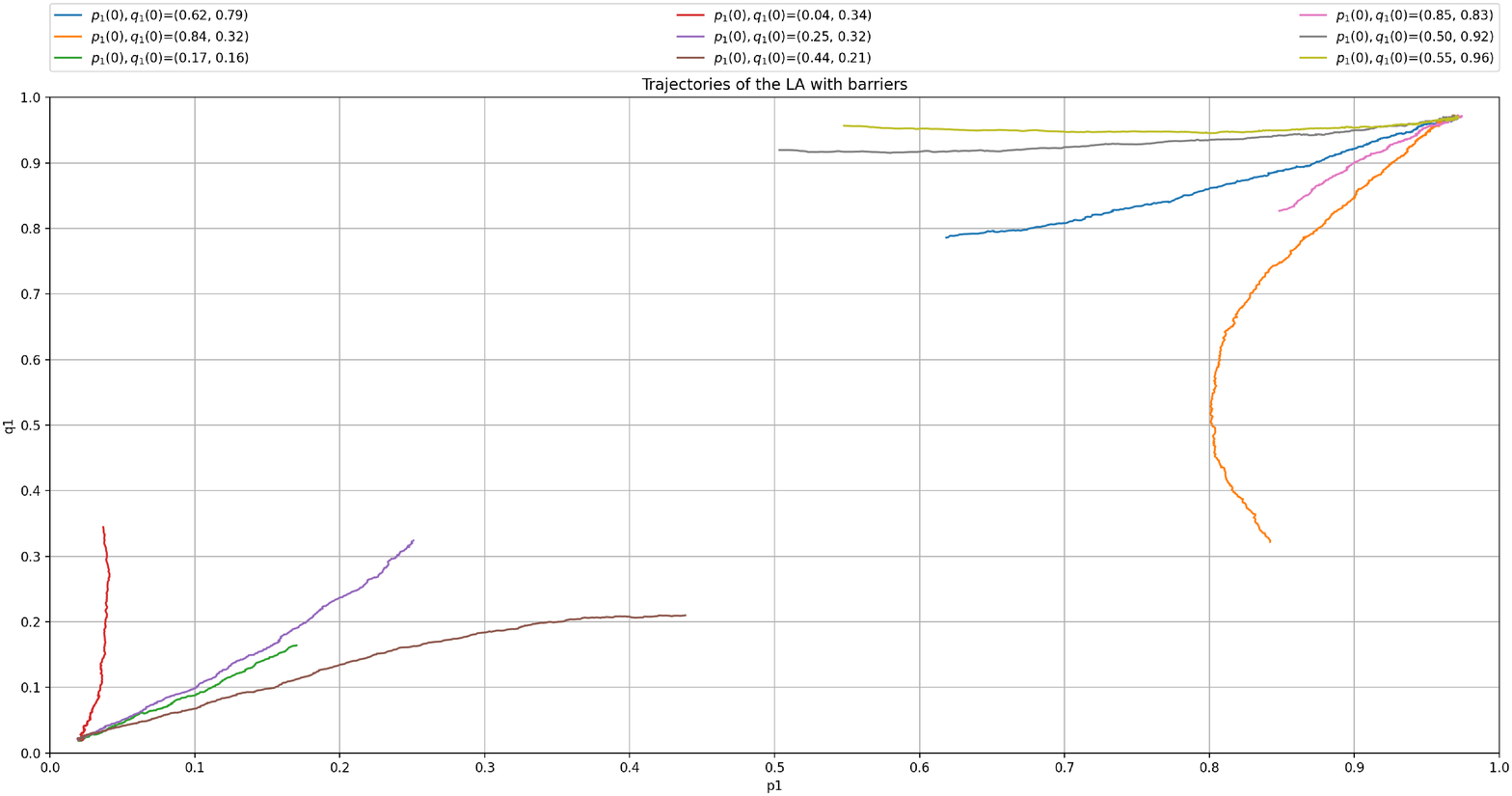}
\caption{9 Trajectories of  the LA with barriers starting from random initial point with $p_{max}=0.99$ and $\theta=0.0001$.}
\label{fig:Case3_trajectories_LA}
\end{figure}

We plot the ODE corresponding to our LA for case 3 in Figure \ref{fig:ODE_two_pure}. We can see two attractions points which approach the two Nash equilibria.

\begin{figure*}[ht]
\includegraphics[scale=0.4]{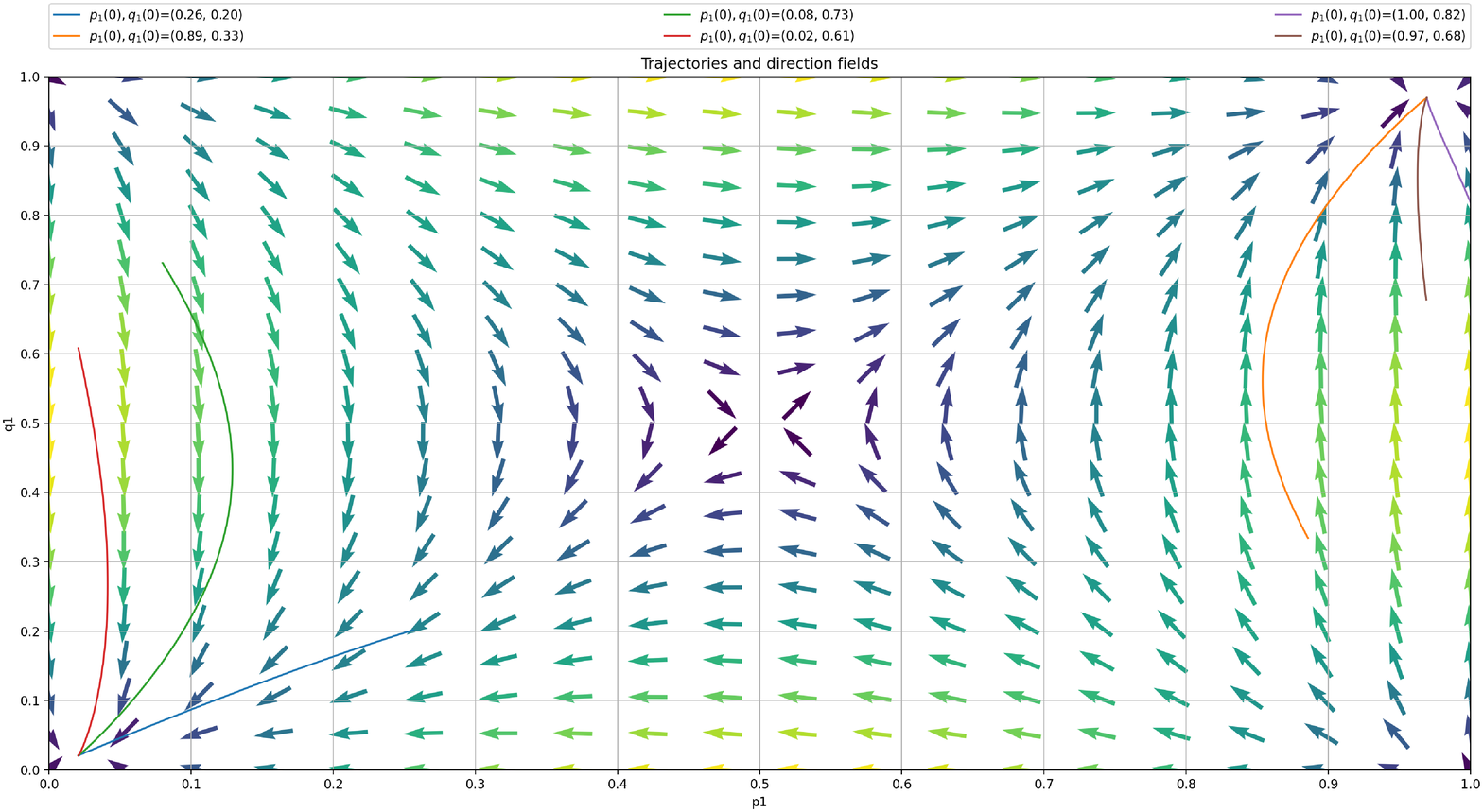}
\caption{Trajectory of ODE  using $p_{max}=0.99$ for case 3.}
\label{fig:ODE_two_pure}
\end{figure*}

\begin{table}[htp!]
  \caption{Error for different values of $\theta$ and $p_{max}$ for the game specified by the $R$ matrix and the $C$ matrix given by Eq. \eqref{eq:R3} and Eq. \eqref{eq:C3}.}
\[
\begin{array}{ | c | c | c | }
\hline
	p_{max} &  \theta=0.0001 &  \theta=0.00001 \\ \hline
	0.990 & 3.08 \times 10^{-2} & 2.06\times 10^{-2} \\ \hline
	0.991 & 2,76 \times 10^{-2} & 2.76 \times 10^{-2} \\ \hline
	0.992 & 1.64 \times 10^{-2} & 2.43 \times 10^{-2} \\ \hline
	0.993 & 1.42 \times 10^{-2} & 2.12 \times 10^{-2} \\ \hline
	0.994 & 1.85 \times 10^{-2} & 1.21 \times 10^{-2} \\ \hline
	0.995 & 0.0 & 1.53 \times 10^{-2} \\ \hline
	0.996 & 0.0 & 1.21 \times 10^{-2} \\ \hline
	0.997 & 0.0 & 0.0 \\ \hline
	0.998 & 0.0 & 0.0 \\ \hline
	0.999 & 0.0 & 0.0 \\ \hline
\end{array}
\]
\label{tab:case3}
\end{table}

Although for $p_{max} \neq 1$, w our scheme  is in theory ergodic and not absorbing, this is not the case in practice as shown in the simulation reported in Table \ref{tab:case3}. In fact for  $\theta=0.0001$ and as $p_{max}$ becomes larger or equal to $0.995$, we observe that the error is zero meaning that the LA has converged already to an absorbing state! This lock in probability phenomenon is due to the limited accuracy of the machine and limitations of the  random number generator. For smaller  $\theta=0.00001$, we expect that the LA will approximate better the ODE. Indeed, this is the case the absorption this time does not happen for $p_{max}=0.995$ and $p_{max}=0.996$ as in the previous case, but happen for only $p_{max}$ larger or equal to $p_{max}=0.997$.

Solving the ODE for $p_{max}=0.999$,  gives two solutions, namely, $(p^*,q^*)=(0.99699397,0.99699397)$ and $(0.00200603,0.00200603)$ which approach $(p_{\mathrm{opt}},q_{\mathrm{opt}})=(1,1)$ and $(p_{\mathrm{opt}},q_{\mathrm{opt}})=(0,0)$ respectively.

Solving the ODE for $p_{max}=0.998$,  gives two solutions $(p^*,q^*)=(0.99397576,0.99397576)$ and $(0.00402424,0.00402424)$.

While solving the ODE for $p_{max}=0.997$,  gives $(p^*,q^*)=(0.99094517,0.99094517)$ and $(0.00605483,0.00605483)$.

\subsection{Real-life Application Scenario}
\label{sec:Applications}

The application of game theory in cybersecurity is a promising research area attracting lots of attention \cite{sokri2020game, fielder2020, do2017}. Our LA-based solution is well suited for that purpose. In cybersecurity, algorithms that can converge to mixed equilibria are preferred over those that get locked into pure ones since randomization reduces an attacker's predictive capability to guess the implemented strategy of the defender. For example, let us consider a repetitive two-person security game comprising of a hacker and network administrator. The hacker intends to disrupt the network by launching a Distributed Denial of Service attack (DDOS) of varying magnitudes that could be classified as high or low. The administrator can use varying levels of security measures to protect the assets.
We can assume that the adoption of a strong defense strategy by the defender has an extra cost compared to a low one. Similarly, the usage of a high magnitude  attack strategy by the attacker has a higher cost compared to a low magnitude attack strategy.
Another example of a security game is the jammer and transmitter game \cite{vadori2015jamming} where 
a jammer is trying to guess the communication channel of the transmitter to interfere and block the communication.
The transmitter chooses probabilistically a channel to transmit over and the jammer choose probabilistically a channel to attack. Clearly converging to pure strategies is neither desirable by the jammer nor by the transmitter as it will give a predictive advantage to  the opponent.

\section{Conclusion}
\label{sec:Conclusions}

Learning Automata (LA) with artificially {\em absorbing}
were first introduced in  1980s \cite{Oommen1986} and have been recently adopted in Estimator LA \cite{zhang2013ACPA, zhang2016formal,yazidi2019hierarchical}. In this paper, we rather propose a LA with artificially {\em non-absorbing} that is able to solve game theoretical problems. The scheme is able to converge  to the game's Nash equilibrium under limited information that has clear advantages over the well-known LA solution for game theoretical due to Sastry et al. \cite{sastry1994decentralized}. In fact, the latter family of solution \cite{sastry1994decentralized} which has found a huge number of applications are only able to only converge to pure strategies and fail to converge to optimal mixed equilibrium. This presents clear disadvantage specially in cases where no Saddle Point exists for a pure strategy and thus the only Nash equilibrium of the game is a mixed one. 
Our scheme is an ergodic one and illustrates a design by which an inherently absorbing scheme, in our case, Linear Reward-Inaction ($L_{R-I}$), is rendered ergodic. Interestingly, while being able to solve the mixed Nash equilibrium case, our scheme maintains the plausible properties of the original $L_{R-I}$ as it is able to converge to a near-optimal  to the pure strategies in the probability simplex whenever a Saddle Point exists for pure strategies.
Furthermore, we also provide a general $S-$type learning scheme for handling continuous feedback and not necessarily binary.
As a future work, we would like to extend our scheme to Stackelberg games which are often employed in security and that assume that the defender deploys a mixed strategy that can be fully observed by the attacker who will optimally reply to it.  The extension would be interesting but from being obvious.


\bibliographystyle{IEEEtran}
\begin{IEEEbiography}[{\includegraphics[width=1in,height=1.25in,clip,keepaspectratio]{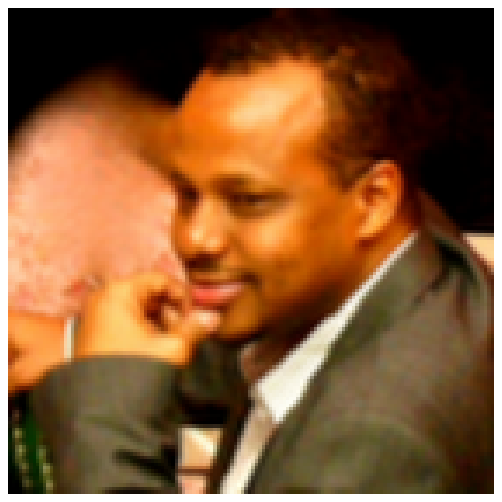}}]{Ismail Hassan}
 received the M.Sc. degree in network and system administration from the University of Oslo, Oslo, Norway, in 2005.
 In 2005, he joined the Oslo University College (OsloMet), Oslo Metropolitan University, Oslo, as a Senior System and a Network Engineer, and after four years, transitioned to the position of Assistant Professor. He is currently an Assistant Professor with OsloMet. His field of interests includes cybersecurity, networking technologies, operating systems, DevSecOps, teaching, and learning methods. 
\end{IEEEbiography}

\begin{IEEEbiography}[{\includegraphics[width=1in,height=1.25in,clip,keepaspectratio]{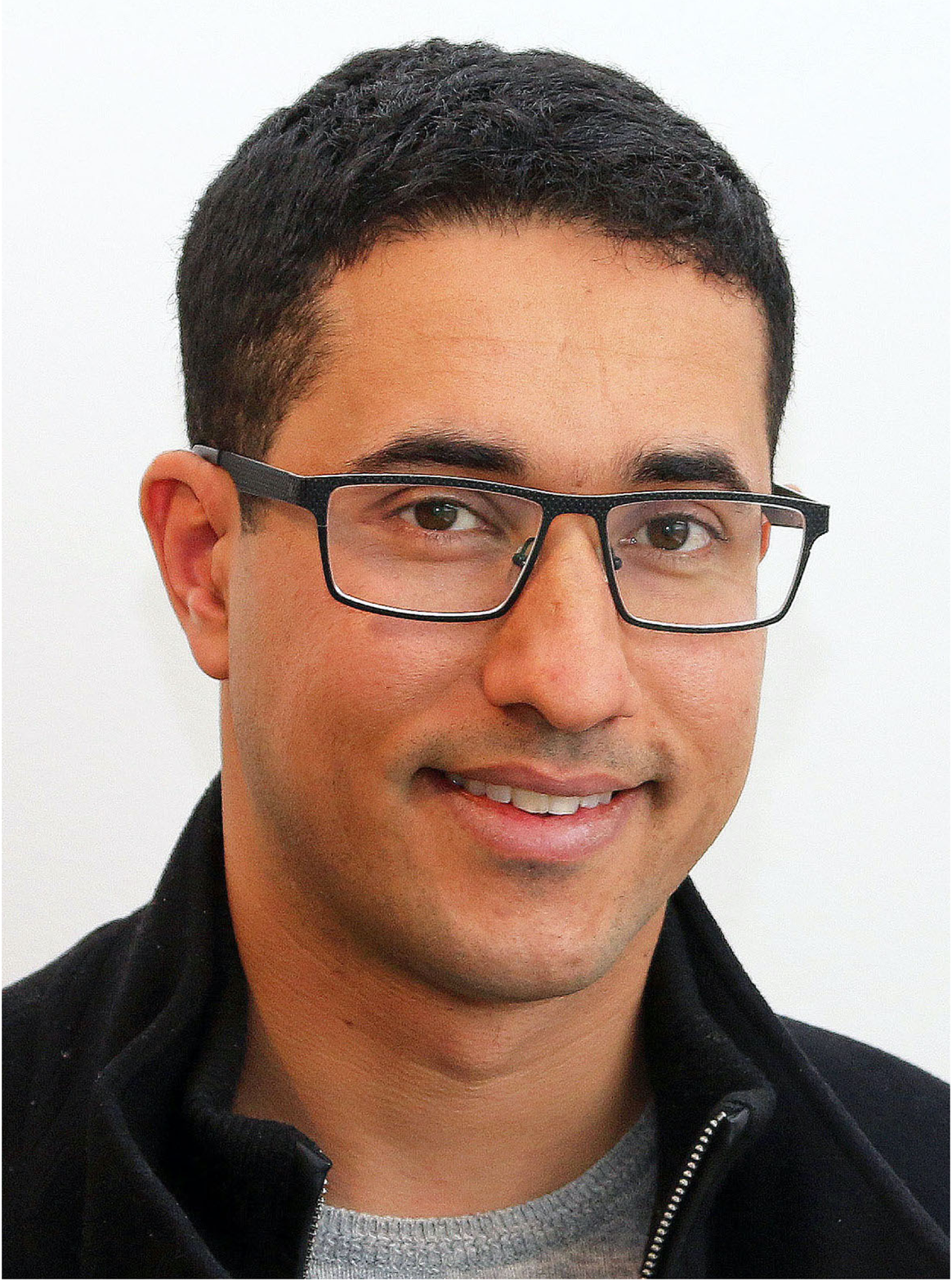}}]{Anis Yazidi}
 received the M.Sc.and Ph.D. degrees from the University of Agder, Grimstad, Norway, in 2008 and 2012, respectively. He was a Researcher with Teknova AS, Grimstad, Norway. From 2014 til 2019 he was an associate professor with the Department of Computer Science, Oslo Metropolitan University, Oslo, Norway.  He is currently a Full Professor with the same departement where he is leading the research group in Applied Artificial Intelligence. He is also Professor II with the Norwegian University of Science and Technology (NTNU), Trondheim, Norway. His current research interests include machine learning, learning automata, stochastic optimization, and autonomous computing.
\end{IEEEbiography}

\begin{IEEEbiography}[{\includegraphics[width=1in,height=1.25in,clip,keepaspectratio]{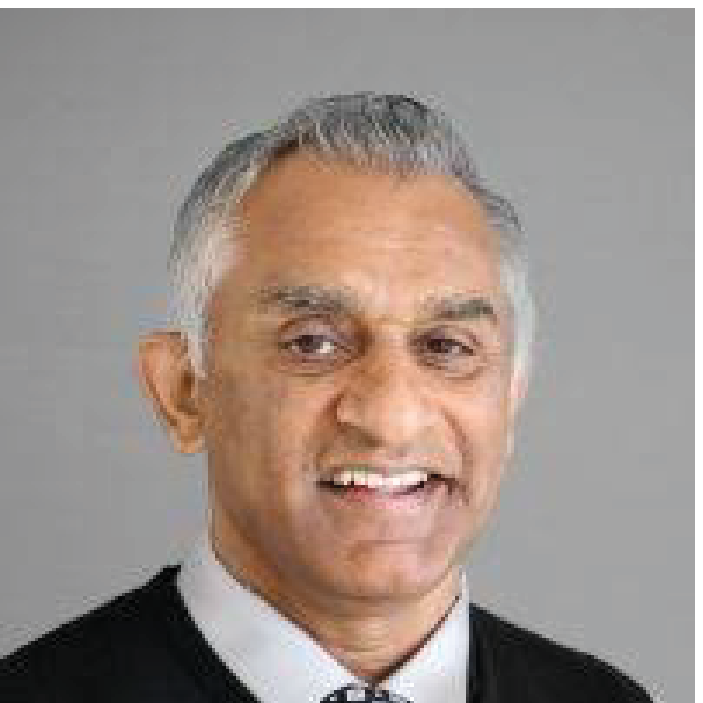}}]{B. John Oommen}
was born in India in 1953. He received his Bachelor of Technology in Electrical
Engineering at the Indian Institute of Technology in Madras, India
in 1975. He then pursued his Master of Engineering degree at the
Indian Institute of Science in Bangalore, India receiving his
degree in 1977. At both these institutions, he won the medal for being
the Best Graduating Student. He received a Master of Science
degree in 1979, and a Ph.D. in Electrical Engineering in 1982,
both from Purdue University, Indiana, USA.

In 2003, Dr. Oommen was nominated as a Fellow of
the Institute of Electrical and Electronic Engineers (IEEE) for
research in a subfield of Artificial Intelligence, namely in Learning Automata. He is currently a Life Fellow of the IEEE.  He was also
nominated as a Fellow of the International Association of Pattern
Recognition (IAPR) in August 2006 for contributions to
fundamental and applied problems in Syntactic and Statistical
Pattern Recognition. He has served on the editorial board of the
journals {\em IEEE Transactions on Systems, Man and Cybernetics}, and
{\em Pattern Recognition}. Dr. Oommen has been teaching in the School of
Computer Science at Carleton University since 1981, and was
elevated to be a Chancellor's Professor at Carleton
University in 2006. He has published more than 485 refereed publications, many of which have been award-winning. He has also won Carleton University's Research Achievement
Award four times, in 1995, 2001, 2007 and 2015.
\end{IEEEbiography}

\begin{appendices}
\section{Norman theorem}
\label{sec:Appendix}

\subsection{Norman theorem}
\label{sec:appendix_A}

\begin{theorem}
\label{thm:Norman}
Let $X(t)$  be a stationary Markov process dependent on a constant parameter $\theta \in [0,1]$. Each $X(t) \in I$, where $I$ is  a  subset  of  the  real  line.  Let $\Delta X(t)=X(t+1)-X(t)$. The following are assumed to hold:

\begin{enumerate}
\item I is compact.
\item $E [\Delta X(t) | X(t)=y]= \theta w(y)+ O(\theta^2)$
\item $Var [\Delta X(t) | X(t)=y]= \theta ^2 s(y)+ o(\theta ^2)$
\item $E [\Delta X(t)^3 | X(t)=y]=  O(\theta ^3)$
where $sup_{y \in I} \frac{O(\theta^k)}{\theta^k}< \infty$ for $K=2,3$ and $sup_{y \in I} \frac{o(\theta^2)}{\theta^2} \rightarrow 0$ as $\theta \rightarrow 0$.
\item $w(y)$ has a Lipschitz derivative in $I$.
\item $s(y)$ is Lipschitz $I$.
\end{enumerate}

If Assumptions (1)-(6) hold, $w(y)$ has a unique root $y^*$ in $I$ and
$\frac{d w}{d y}  \bigg|_{y=y^*} \le 0$ then
\begin{enumerate}

\item $var [\Delta X(t) | X(0)=x]=0(\theta)$ uniformly for all $x \in I$ and $t \ge 0$.
For any $x \in I$,  the differential equation $\frac{d y(\tau)}{d \tau}=w(y(t))$ has a unique solution  $y(\tau)=y(\tau,x)$ with $y(0)=x$  and	$E [\delta X(t) | X(0)=x]=y( t \theta)+O(\theta)$ uniformly for all $x \in I$ and  $t \ge 0$.

\item $\frac{X(t)-y(t \theta)}{\sqrt \theta}$ has a normal distribution with zero mean and finite variance as $\theta \rightarrow 0$ and $t \theta \rightarrow \infty$.

\end{enumerate}

\end{theorem}
\newpage
\section{Experimental results for $S$-type Environments}
\label{sec:appendix_B}

In this section, we present the results of the experiments for the S-type learning
game. We conducted several simulations similar to those presented in Section \ref{sec:simualations_L_RI}. The same instances of the payoff matrices $R$ and $C$ were used, covering the cases referred to in Section \ref{sec:scheme_L_R_I}. 	

For all the experiments conducted for the $S$-LA, 9 trajectories were plotted for 2,000,000 iteration, with $p_{max}=0.99$ and $\theta=0.0001$.  A general observation that we noticed when performing that the experiments is that the $S$-LA converges slower than the the $L_{R-I}$. Therefore, we have doubled the number of iterations to allow the $S$-LA to converge in our experiments.

\subsection{Case 1: Only one mixed Nash equilibrium exists}
\label{sec:SLearning-case1}
Figure \ref{fig:S_model_case1} depicts the situation where the only Nash equilibrium that exists is a mixed one. We can easily observe that the $S$-LA approaches the trajectories of the ODE given in Figure \ref{fig:ODE_Mixed}. Please note that the ODE regardless of the LA type, whether it is $L_{R-I}$ or $S-$LA.

\begin{figure}[htp!]
\centering
\includegraphics[width=8cm]{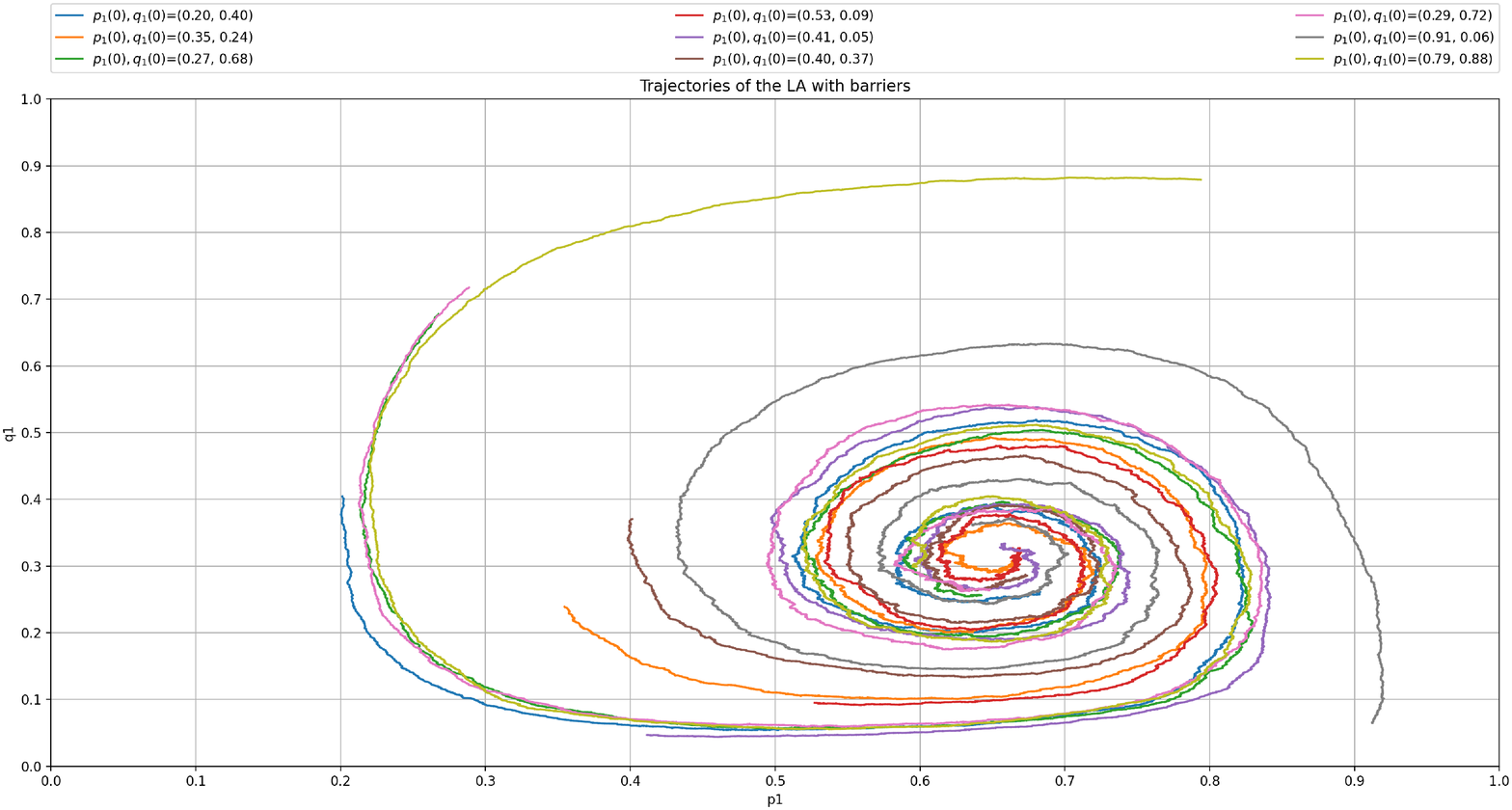}
\caption{Trajectory of $S$-LA using $p_{max}=0.99$ and $\theta=0.0001$ for case 1.}
\label{fig:S_model_case1}
\end{figure}

\subsection{Case 2: One Pure equilibrium}
\label{sec:SLearning-case2}
We also examined the case where the game has a single pure equilibrium. The exhibited behavior is comparable to those reported in Section \ref{sec:simualations_L_RI}. The trajectory of the LA depicted in Figure \ref{fig:S_model_case2} tightly follows the trajectories of the ODE depicted in Figure \ref{fig:ODE_one_pure_099}.
As $\theta$ goes to zero, the trajectories of the LA and those of the ODE will be indistinguishable 
\cite{sastry1994decentralized}.

\begin{figure}[htp!]
\centering
\includegraphics[width=8cm]{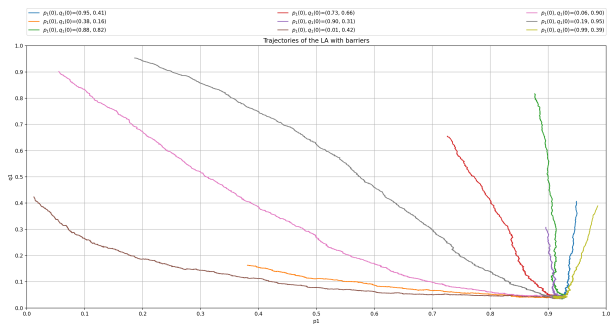}
\caption{Trajectory of $S$-LA using $p_{max}=0.99$ and $\theta=0.0001$ for case 2.}
\label{fig:S_model_case2}
\end{figure}

\subsection{Case 3: Two Pure equilibria and one mixed}
\label{sec:SLearning-case3}
Figure \ref{fig:S_model_case3} shows the situation where there are two pure equilibria and one mixed.

We observe that the LA converges to one of the two pure equilibria that is closest to the starting point. The $S$-LA behaves much similar to the $L_{R-I}$ LA as shown in Figure \ref{fig:Case3_trajectories_LA}. We also observe that the $S$-LA respectively converged to the Nash equilibrium close to (1, 1) and close to (0,0) approximately 50\% of the time.

\begin{figure}[htp!]
\centering
\includegraphics[width=8cm]{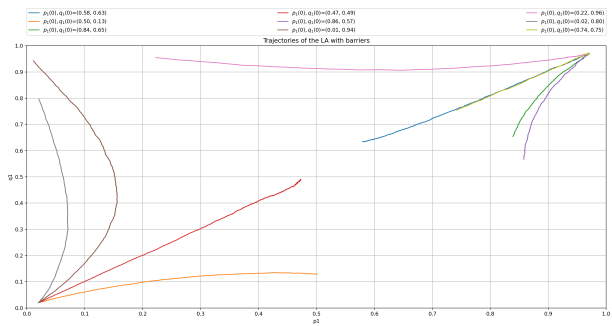}
\caption{Trajectory of $S$-LA using $p_{max}=0.99$ and $\theta=0.0001$ for case 3.}
\label{fig:S_model_case3}
\end{figure}
\end{appendices}
\end{document}